\documentclass{lmcs} 
\pdfoutput=1

\usepackage{lastpage}
\lmcsdoi{22}{3}{7}
\lmcsheading{}{\pageref{LastPage}}{}{}%
{Jul.~03,~2025}{Aug.~18,~2026}{}


\usepackage[utf8]{inputenc}
\usepackage{amssymb, latexsym, stmaryrd, amsthm, dsfont, amsfonts, amsbsy, amsthm, amsmath, mathrsfs}
\usepackage{tikz-cd}
\usepackage{mathtools}
\usepackage{bm}
\usepackage{enumitem}
\usepackage{hyperref}

\let\bigland=\bigwedge
\let\biglor=\bigvee
\let\isom=\cong
\newcommand{\Luk}{\text{\textsl{\L}}}

\newcommand{\und}{\mathord{\reflectbox{/}}}

\let\phi=\varphi

\newcommand{\CM}{\mathcal{M}}

\newcommand{\CL}{\mathcal{L}}

\newcommand{\CP}{\mathcal{P}}

\begin{document}

\title[Residuated Lattice Locality]{Locality in Residuated Lattice Structures}

\author[J.~Carr]{James Carr\lmcsorcid{0000-0002-2746-2446}} 

\address{University of Queensland}
\email{james.carr.47012@gmail.com}  

\begin{abstract}
\noindent Many-valued models generalise the structures from classical model theory by defining truth values for a model with an arbitrary algebra. Just as algebraic varieties provide semantics for many non-classical propositional logics, models defined over algebras in a variety provide the semantics for the corresponding non-classical predicate logics. In particular, models defined over varieties of residuated lattices represent the model theory for first-order substructural logics.

In this paper we study the extent to which the classical locality theorems from Hanf and Gaifman hold true in the residuated lattice setting. We demonstrate that the answer is sensitive both to how locality is understood in the generalised context and the behaviour of the truth-defining algebra. In the case of Hanf's theorem, we will show that the theorem fails for the natural understanding of local neighbourhoods, but is recoverable with an alternative understanding  for well-connected residuated lattices. For Gaifman's theorem, rather than consider Gaifman normal forms directly we focus on the main lemma of the theorem from textbook proofs - that models which satisfy the same basic local sentences are elementarily equivalent. We prove that for a number of different understandings of locality, provided the algebra is well-behaved enough to express locality in its syntax, this main lemma can be recovered. In each case we will see the importance of an order-interpreting connective which creates a link between the modelling relation for models and formulas and the valuation function from formulas into the algebra. This link enables a syntactic encoding of back-and-forth systems providing the main technical ingredient to proofs of the main locality results. 
\end{abstract}

\maketitle


\section{Introduction}\label{sec:Intro}
An important limitation of classical first-order logic is its \textit{locality}. Whether a first-order formula is modelled in a given structure depends only on the behaviour of the model at a `local' level, with the consequence that first-order logic cannot express `global' properties. This is at its most transparent in graphs where the natural notion of distance between points makes precise the idea of `local'; locality results provide an effective method to show some natural graph concepts such as connectivity or acyclicity are inexpressible in logical terms \cite[Section 3.6]{Libkin04:Elementsfinmod}. In arbitrary signatures locality is made precise through the \textit{Gaifman graph} of a model, allowing the natural notion of distance in a graph to define one for any structure \cite[Section 1.4]{EbbinghausFlum95:Finmodeltheory}. Once distance is understood, a multitude of locality results can be expressed, the two most significant being \textit{Hanf locality} \cite{Hanf65:Modeltheory} and \textit{Gaifman locality} \cite{Gaifman81:Local}. Hanf locality links the equivalence of sentences up to a given quantifier rank to the number of isomorphic local substructures; Gaifman locality states that any first-order formula is equivalent to one restricted to `local quantification'. Locality results are especially interesting in the context of finite model theory due to the topic of database queries for which locality results provide an effective method of proving a given query is not first-order definable \cite{Libkin04:Elementsfinmod}. 

The utility of locality results in classical finite model theory motivates understanding how similar theory could be developed in other model theoretic contexts. Previous work in this vein comes from Bizi\`{e}re, Gr\"{a}del and Naaf \cite{BiziereGradelNaaf23:Localitysemiring} who investigated locality for finite models defined over ordered semirings, and from Wild, Schr\"{o}der, Pattinson and K\"{o}nig \cite{WildSchroderPattinsonKonig18:fuzzyvanBenthem} who remarked on locality for the first-order fragment corresponding to a fuzzy modal logic. Here we aim to do the same for finite models defined over \textit{residuated lattices}. Varieties of residuated lattices provide the algebraic semantics for a multitude of many-valued logics encompassing many of the most well-studied \textit{substructural logics} \cite{GalatosJipsenKowalskiOno07:Resilattice}. Specific varieties have been the setting of much work regarding non-classical model theory, including the development of back-and-forth criteria \cite{DellundeGarciaNoguera18:Fuzzybackandforth}, the status of amalgamation theorems \cite{BadiaCostaDellunde19:Synccharfuzzy} or 0-1 laws \cite{BadiaCaicedoNoguera25:01law}. Notably, neither the semiring or residuated lattice setting subsumes the other. Semiring semantics takes standard classical syntax and generalises the behaviour of the conjunction and disjunction connective but lacks a primitive notion of implication or negation. The result is that our investigation develops quite differently from that of Bizi\`{e}re et al. 

Our choice to work with residuated lattices is far from arbitrary, rather it reflects a common pattern of research in non-classical logic reflected in our investigation. Namely, we identify a topic from classical model theory, adjust the core concepts for the generalised setting and attempt to recover a proof in as weak a setting as possible. As problems arise the behaviour of models is strengthened, adopting the main features of classical logic used in the original proof. Ideally once a theorem is recovered those features are shown to be necessary. In this sense to study non-classical model theory is also to study classical model theory identifying the exact critical features of classical models that underlie a particular theorem. When it comes to locality, an ongoing theme will be the importance of the presence of an order-defining connective, that is an algebraic term $\tau(x,y)$ that takes a `true' value exactly when $x$ is less than or equal to $y$ in the lattice order. In residuated lattices this is a critical property of the left and right residuums of the strong conjunction, or in the commutative case the implication connective $\rightarrow$. The presence of this connective creates a link between the usual modelling relation $\models$ and the valuation function $||-||$ which constitute the two central relationships between syntax and semantics, and from this link the more substantive syntax-semantic relationships of isomorphism types and locality theorems are then recovered. 

Exploring these questions in the residuated lattice setting the initial question is how to give generalisations of the concepts needed to express locality. As is not uncommon in many-valued model theory, there are multiple sensible generalisations united in some basic properties and relationships. In the case of locality, what emerges is that those basic properties and relationships are sufficient to recover much of the theory, with the proof work itself following through as it does classically. This determines the shape of our investigation. In Section~\ref{sec:Prelim} we properly introduce many-valued models as well as the algebras they are defined relative to. We then briefly detail some of the critical model theory for these models that will be necessary for the main investigation, including highlighting the particularly important subclasses of standard expansions and witnessed models. Our main investigation is split into two sections, Section~\ref{sec:Hanf} focusing on Hanf locality and Section~\ref{sec:Gaif} on Gaifman locality. In each case we establish a variant of the locality results for models defined over a residuated lattice with some additional constraints. We then compare and contrast our work with residuated lattices with the investigation for semirings in Section~\ref{sec:semirings}. Finally, in Section~\ref{sec:queries} we briefly see how our locality results can be used in similar manner to the classical results to provide a method of showing queries for the considered models are not definable, before concluding with some remarks about further study. 

As in the classical setting, locality is investigated only for relational languages, possibly with object constants \cite{EbbinghausFlum95:Finmodeltheory,Libkin04:Elementsfinmod}. Furthermore, as is common in the many-valued setting, we restrict to models defined over a fixed algebra (see also \cite{BadiaCostaDellunde19:Synccharfuzzy,DellundeVidal19:ppform}).


\section{Preliminaries}\label{sec:Prelim}
We start with a brief introduction to our objects of study - many-valued models. Although we will focus solely on finite models in our main investigation we introduce our models in a general manner. Many-valued models consist of two components, a non-empty set equipped with interpretations of non-logical symbols akin to familiar classical models, and an algebra that is used to define truth. Naturally, the properties of structures are tied to the choice of algebras they are defined over. Different motivations lead one to be concerned with different classes of models. The form of our investigation is an attempt to understand for which truth-defining algebras some interesting model theory can be recovered - here the theory in question is locality. We choose to work at a level of generality that can be applied cleanly to a wide range of more specific motivated settings whilst also making transparent the essential properties the locality results are sensitive to. 

\begin{defi}
    A \textit{residuated lattice} is an algebra $A$ in signature $\CL=\langle \land,\lor,\cdot,\und,/, 1\rangle$ such that:\begin{itemize}[nosep]
        \item $\langle A,\land,\lor\rangle$ is a lattice.
        \item $\langle A,\cdot,1\rangle$ is a monoid
        \item The residuation properties hold for $\langle A,\cdot,\und,/\rangle$; for all $a,b,c\in A$\[a\leq c/b\text{ iff }a\cdot b\leq c\text{ iff }b\leq a\und c.\]
    \end{itemize}
    A residuated lattice $A$ is \textit{well-connected} iff for all $a,b\in A$, $a\lor b\geq 1$ iff $a\geq 1$ or $b\geq 1$. It is \emph{non-trivial} iff it contains atleast two elements.  
\end{defi}

Residuated lattices provide an interesting setting for many-valued model theory as they provide the algebraic semantics for \textit{substructural logics}, a well established class of propositional logics containing most of the well-studied particular logics in the literature, including Boolean algebras, Heyting algebras, MV-algebras and lattice ordered groups \cite[Section 3.4]{GalatosJipsenKowalskiOno07:Resilattice,CintulaHajekNoguera11:Handbookvol1}. The models defined over these algebras provide the semantics to the first-order predicate logics extending their respective propositional logic \cite[Chapter 7]{CintulaNoguera21:LogicImplication}. The consideration of well-connected lattices will be essential to have the appropriate behaviour of the existential quantifier and is partially justified due to all (finitely) subdirectly irreducible (F)SI residuated lattices being well-connected \cite[Lemma 5.29]{GalatosJipsenKowalskiOno07:Resilattice}. This gives the logical viewpoint of our investigation, to what extent does the locality of classical logic extend to substructural logics. 

We can immediately highlight the basic properties of well-connected residuated lattices that motivate their use for the algebras in our models. Firstly is that the residual operators interpret the algebraic order; \[1\leq b/a\text{ iff }a\leq b\text{ iff }1\leq a\und b.\] Secondly is that our lattice connectives are stable around $\geq 1$; \begin{align*}
    a\land b\geq 1 \text{ iff } a\geq 1 \text{ and }b\geq 1 &&
    a\lor b\geq 1 \text{ iff }a\geq 1 \text{ or }b\geq 1.
\end{align*}

During our investigation we will identify two subclasses of well-connected residuated lattices important to locality. The first are those which are bounded. We say that a (well-connected) residuated lattice $A$ is \emph{bounded below} iff there is an element $\bot$ such that $\forall a\in A\ \bot\leq a$, and \emph{bounded above} iff there is an element $\top$ such that $\forall a\in A\ a\leq \top$. We say that $A$ is \emph{bounded} iff it is both bounded below and bounded above. A basic property of residuated lattices is that bounded below implies bounded, as $\bot$ is an annihilator for $\cdot$ and consequently $A$ is bounded above by $\bot\und\bot=\bot/\bot$ \cite[Section 2.2]{GalatosJipsenKowalskiOno07:Resilattice}. 

\begin{exa}\label{exa:wellconnectboundresi}
Any residuated lattice defined on the unit interval is well-connected and bounded. In the literature on many-valued logic these are commonly called \emph{standard algebras}. Particular examples include \textit{the} standard algebras for the three fuzzy logics; \L ukasiewicz $\Luk$, G\"{o}del $\mathrm{G}$ and product $\Pi$. The monoid operators $\cdot$ of these algebras are all commutative and so their left and right residuums coincide and denoted by the single connective $\rightarrow$. The two connectives are defined as follows; the definition of $\rightarrow$ given is when $x>y$ with $x\rightarrow y=1$ otherwise.\begin{align*}
    x&\cdot_{\Luk}y = max\{0,x+y-1\} & x&\cdot_{\mathrm{G}}y = min\{x,y\} & x&\cdot_{\Pi}y = xy\\
    x&\rightarrow_{\Luk}y = min\{1,1-x+y\} & x&\rightarrow_{\mathrm{G}}y = y & x&\rightarrow_{\Pi}y = y/x 
\end{align*}

All the above lattices are \emph{integral}, i.e. the unit $1$ of the monoid operator $\cdot$ is the top element in the lattice order. Integrality is not a requirement for our investigation, an example of a non-integral standard algebra is the standard algebra based on the involutive left-continuous conjunctive uninorm $\ast$ defined by \cite[Section 2.1]{CintulaHajekNoguera11:Handbookvol1}:\[x\ast y=\begin{cases}
    min\{x,y\} & \text{if }y\leq 1-x,\\
    max\{y,x\} & \text{ otherwise}.
\end{cases}\]
 
Whilst each of the preceding examples have a linear underlying lattice order we we can give non-linear examples. Indeed, each of the preceding examples are (finitely) subdirectly irreducible, and any such residuated lattice is well-connected~\cite[Section 5.3]{GalatosJipsenKowalskiOno07:Resilattice}. In particular, FSI Heyting algebras provide a ready class of bounded well-connected residuated lattices with a (potentially) non-linear lattice order.  
\end{exa}

The second subclass does restrict to algebras whose underlying lattice order is linear and moreover have a special element just below the unit element. For a residuated lattice $A$, a \emph{co-atom} is an element $c\in A\setminus{1}$ such that for all $a\in A\ a\leq c$ iff $1\not\leq a$.\footnote{The name is taken from the integral case where $c$ is a co-atom in the familiar sense.} A residuated \emph{chain} is a residuated lattice whose lattice order is linear. 

\begin{exa}\label{exa:resichaincoatom}
    The most natural examples of residuated chains with co-atoms are those defined on finite chains relative to the fuzzy logics $\Luk_n$ and $\mathrm{G}_n$, that is the $n$-element MV-chains and $n$-element G\"{o}del-chains \cite[Example 2.4.2, Definition 1.3.1 ]{CintulaHajekNoguera11:Handbookvol2}. Explicitly, these are defined on the $n$-element set $\{k/n\in [0,1]:0\leq k\leq n\}$ with $\cdot$ and $\rightarrow$ defined as in $\Luk$ and $\mathrm{G}$ in the previous example.
    
    In this case we do not require the algebra to be bounded. Natural examples of non-bounded residuated chains with co-atoms are finite Wajsberg hoops, that is commutative integral residuated lattices that also satisfy \emph{Tanaka's equation} and \emph{divisibility}~\cite{Ugolini34:Wajsberghoops}.
\end{exa}

With an overview of the algebraic part of our models, let us now introduce our basic syntax and semantics. 
\begin{defi}
    A \textit{predicate language} $\CP$ is a pair $\langle F,R\rangle$ where $F$ is a set of \textit{function symbols} and $R$ a set of \textit{relation symbols} (disjoint from $F$) all of which come with an assigned natural number - its arity. The relation symbols of arity zero are called \textit{truth constants} and function symbols of arity zero are called \textit{object constants}. 

    Let $\CL$ be a propositional/algebraic signature $\CL$ and $\mathrm{Var}$ a countable set of variables. We define the sets of $\CL,\CP$\textit{-terms}, $\CL,\CP$\textit{-atomic formulas} $\CL,\CP$\textit{-formulas} inductively just as in classical logic except we naturally insist that the formulas are closed under all connectives from $\CL$. We denote the set of $\CP$-formulas by $\mathrm{Fm}(\CP)$. We similarly define notions such as free occurrence of a variable, open formula, substitutability and sentence as in classical model theory. 

    We always assume that we have the equality symbol $=$ in our predicate language. 
\end{defi}
\begin{rem}
    As the algebraic signature is usually clear from context we frequently omit reference to it and talk of $\CP$-terms, formulas, sentences etc. Indeed, throughout our main investigation our algebraic signature is that of residuated lattices $\CL=\{\land,\lor,\cdot,\und,/,1\}$.
\end{rem}

\begin{defi}\label{def:structures}
    Let $\CP$ be a predicate language. A $\CP$\textit{-structure} is a pair $\langle A,\CM\rangle$ where $A$ is a (non-trivial) residuated lattice and $\CM$ is a triple $\langle M,\langle F^M\rangle_{F\in\CP},\langle R^M\rangle_{R\in\CP}\rangle$ where:\begin{itemize}
        \item $M$ is a non-empty set.
        \item $F^M\colon M^n\rightarrow M$ for each $n$-ary function symbol $F\in\CP$ ($F^M\in M$ when $n=0$). 
        \item $R^M\colon M^n\rightarrow A$ for each $n$-ary relation symbol $R\in\CP$ ($R^M\in A$ when $n=0$).
    \end{itemize} 
    The set $M$ is called the \textit{domain} of $\CM$ and the mappings $F^M,R^M$ the interpretations of function and predicate symbols in $\CM$.
\end{defi}
\begin{rem}
    The insistence on $A$ being non-trivial is made simply because models defined over 1-element algebras are uninteresting. Moreover, for languages with equality it is needed to have a coherent interpretation of the equality symbol. Further, observe that when $A$ is the $2$-element Boolean algebra we recover the standard definition of classical model theory. This is true for all of our basic definitions to follow. 
    
    We frequently drop reference to $A$ and speak simply of $\CP$-models $M$ either when $A$ is clear from context or when we are restricting attention to models defined over the same fixed algebra. In this case we frequently use $M$ interchangeably for the model and its underlying domain. 

    In definition \ref{def:structures} we see an important difference between propositional constants $d\in\CL$ and interpreted nullary relation symbols $R\in\CP$. The former will be interpreted by the same truth value $d\in A$ for each structure $(A,M)$ whereas the latter is interpreted by some truth value $P^M\in A$ in each structure separately.
\end{rem}
 
\begin{defi}\label{def:valuation}
    Let $\CM$ be a $\CP$-structure. We define an $M$-valuation as a mapping $v\colon \mathrm{Var}\rightarrow M$. Given an $M$-valuation $v$, an object variable $x$ and an element $m\in M$ we denote by $v_{x=m}$ the $M$-valuation defined as:\[v_{x=m}(y):=\begin{cases}
        m & \text{if }y=x\\
        v(y) & \text{otherwise}
    \end{cases}\]
    We extend an $M$-valuation $v$ to all $\CP$-terms by setting for each $n$-ary $F\in\CP$: \[v(F(t_1,...,t_n))=F^M(v(t_1),...,v(t_n)).\]
    The \textit{valuation function} $||-||\colon \mathrm{Fm}(\CP)\rightarrow A$ of the $\CP$-formulas for a given $M$-valuation $v$ is the (partial) function defined recursively as follows, where the element $e\in A$ used to define equality is an arbitrary element not equal to the unit 1.\begin{align*}
        ||t_1=t_2||^M_v&=
        \begin{cases}
            1 &\text{if } v(t_1)=v(t_2)\\
            e &\text{if } v(t_1)\not=v(t_2)
        \end{cases}\\
        ||R(t_1,..,t_n)||^M_v &=
        R^M(v(t_1),...,v(t_n)),\text{ for each }n\text{-ary }R\in \CP,\\
        ||\circ(\phi_1,...,\phi_n)||^M_v &=\circ^A(||\phi_1||^M_v,...,||\phi_n||^M_v),\text{ for each }n\text{-ary }\circ\in\CL,\\
        ||\forall x\phi(x)||^M_v &=\inf_{\leq}\{||\phi(x)||^M_{v_{x=m}}:m\in M\},\\
        ||\exists x\phi(x)||^M_v &=\sup_{\leq}\{||\phi(x)||^M_{v_{x=m}}:m\in M\}.
    \end{align*} 
    If the infimum or supremum does not exist we take the truth-value of such a formula to be undefined. We say that a $\CP$-structure is \textit{safe} iff $||\phi||^M$ is defined for every formula and reserve the term $\CP$-model for safe $\CP$-structures. 
\end{defi}
\begin{rem}
    The $e\in A\setminus\{1\}$ to define equality is only important insofar as it means equality is \textit{crisp}. In more specific circumstances there is usually a natural value for inequality to take, e.g. when $A$ is bounded $e$ would be taken as the lower bound $\bot$. 
    
    As often in the case of model theory, we usually suppress the direct consideration of the valuation $v$ and write as if the function $||-||^M$ is acting directly on $\CP\cup M$ sentences, e.g. writing $||R(m_1,...,m_n)||^M$ in place of $||R(t_1,...,t_n)||^M_v$ where $||t_i||^M_v=m_i$. 

    We are exclusively interested in safe $\CP$-structures, hence the reservation of the term models; indeed any finite $\CP$-structure defined over a lattice is safe really eliminating the need for the distinction for finite structures. The distinction is mostly relevant when defining (infinite) $\CP$-structures where we need to be careful to ensure they are safe, frequently this is done by choosing $A$ to be a complete lattice as then again any $\CP$-structure defined over $A$ would be safe.
\end{rem}

The valuation function $||-||^M\colon \mathrm{Fm}\rightarrow A$ is really the central relationship between syntax and semantics; the many-valued counterpart to the classical $\models$ relation. Indeed, many model theoretic definitions can be recast with the valuation function.
\begin{defi}
    Let $M,N$ be two $\CP$-models defined over the same (well-connected) residuated lattice $A$. We say that $M,N$ are strongly equivalent, denoted $M\equiv^s N$ iff for all $\phi\in \mathrm{Fm}\ ||\phi||^M=||\phi||^N$.  

    Defining the quantifier depth ($qd(-)$) of a $\CP$-formula exactly as in classical logic we say that $M,N$ are $k$-strongly equivalent, denoted $M\equiv^s_k N$ iff for all $\phi\in \mathrm{Fm}:qd(\phi)\leq k\ ||\phi||^M=||\phi||^N$. 
\end{defi}

\begin{defi}
    Let $M,N$ be two $\CP$-models defined over the same algebra $A$. A map $f\colon M\rightarrow N$ is a \textit{strong homomorphism} iff:\begin{align*}
        &\text{for every }F\in\CP\text{ and }\bar{m}\in M\ f[F^M(\bar{m})]=F^N(f(\bar{m})).\\
        &\text{ for every }R\in\CP\text{ and }\bar{m}\in M\ R^M(\bar{m})= R^N(f(\bar{m})).
    \end{align*}
    \textit{Embeddings} are injective strong homomorphism and \textit{isomorphisms} are bijective strong homomorphisms. We say that $N$ is a \textit{substructure} of $M$ iff there is an embedding $\iota\colon N\hookrightarrow M$. 
\end{defi}

One can investigate these models as independent mathematical objects, working solely with the valuation function.\footnote{Examples of investigations with this focus are \cite{BiziereGradelNaaf23:Localitysemiring} and \cite{HorcikMoraschiniVidal17:Algconstraintsat}.} However, the logical viewpoint motivates considering a modelling relation $\models$ linking syntax and semantics in a less fine grained way. 

\begin{defi}
    Let $A$ be a (well-connected) residuated lattice. We say that a $\CP$-model $\CM=(A,M)$ is a model of a $\CP$-sentence $\phi$ iff $||\phi||^M\geq 1^A$. We also say that $\phi$ is \textit{valid} in $M$. 

    Given a class of residuated lattices $\mathbb{K}$ the $\mathbb{K}$ \textit{based predicate logic} in language $\CP$ is the relation $\models^{\CP}_{\mathbb{K}}$ defined for $\Gamma\cup\{\phi\}\subseteq \mathrm{Fm}(\CP)$ by:\[\Gamma\models^{\CP}_{\mathbb{K}}\phi\text{ iff for each }A\in\mathbb{K}\text{ and each }\CP\text{-model }(A,M), (A,M)\models\Gamma\text{ implies }(A,M)\models \phi.\]
\end{defi}
\noindent
The question of axiomatization for the semantically defined predicate logics $\models_{\mathbb{K}}$ is a considerable topic in its own right, one which is naturally quite sensitive to the properties of the subset $\mathbb{K}$.\footnote{For a detailed discussion one can consult \cite[Chapter 7]{CintulaNoguera21:LogicImplication}.} In our investigation our restriction to considering models defined over a fixed algebra $A$ would effectively mean taking the subset $\mathbb{K}$ as a singleton. In practice, our interest is really with the models as objects themselves and in particular the relationship between the $\models$ relation and strong equivalence $\equiv^s$. Indeed, the modelling relation recasts the essential properties of our underlying algebras; the critical link is the behaviour of the lattice connectives $\land$ and $\lor$ and the residual operators $\und$ and $/$ with respect to $\models$. \begin{align*}
    M\models\phi\land\psi \text{ iff }&M\models\phi \text{ and } M\models\psi.\\
    M\models\phi\lor\psi \text{ iff }&M\models\phi\text{ or }M\models\psi.\\
    M\models\phi\und\psi \text{ iff }||\phi||^M&\leq ||\psi||^M \text{ iff }M\models\psi/\phi. 
\end{align*}
Notably, in the case of finite models these extend to the behaviour of $\forall$ and $\exists$:\begin{align*}
    M\models \forall x\phi(x)&\text{ iff for all }m\in M\ M\models\phi(m).\\
    M \models\exists x\phi(x)&\text{ iff for some }m\in M\ M\models\phi(m).
\end{align*}

This points towards a more `axiomatic' version of our investigation; for any class of models defined over an algebra $A$ in signature $\CL$ we could proceed on the assumption that there are some algebraic terms $\tau_{\land}(x,y), \tau_{\lor}(x,y)$ and $\tau_{\leq}(x,y)$ with the behaviour of $\land,\lor$ and $\und (/)$ relative to a given $\models$ relation described above. We will however proceed in a more concrete setting.

There are two subclasses of models that are important for our investigation. The first are the witnessed models.
\begin{defi}
    We say that a model $(A,M)$ is $\exists$\textit{-witnessed} iff for all formulas $\exists x\phi(x)$ there is an element $m\in M$ for which $\phi(m)$ achieves the value of the supremum, i.e. $||\exists x\phi(x)||^M=||\phi(m)||^M$. We define $\forall$-\emph{witnessing} similarly and say that $M$ is \textit{witnessed} when it is both $\exists$ and $\forall$-witnessed. 
\end{defi}
\begin{rem}\label{rem:witnessedifflinear}
    In the case of finite models $\exists$ and $\forall$ essentially become shorthand for strings of $\land$ and $\lor$ and thus whether a model is witnessed is determined by their behaviour. In particular, finite models defined over linearly ordered algebras (chains) are always witnessed. Conversely, if \textit{every} (finite) $\CP$-model over an algebra $A$ is witnessed then $A$ is a chain. 
\end{rem} 

The second are the standard expansions of models where we add fresh relation symbols to our language and whose intended interpretation are the truth values of the algebra. This lets us introduce syntactic truth constants at the predicate language level rather than in the algebraic signature. 
\begin{defi}
    Let $A$ be a well-connected residuated lattice and $\CP$ a predicate language. We expand the language $\CP$ to one containing a $0$-ary relation symbol for each element $a\in A$, i.e. $\CP^\ast:=\CP\cup A$. For each $\CP$-model $M$ defined over $A$ its \textit{standard expansion} $M^\ast$ is the $\CP^\ast$-model for which $||a||^M=a$ for all $a\in A$. 
\end{defi}
\begin{rem}
    A common proof strategy in many-valued model theory is to work with the standard expansions of models before pulling the results back to the $\CP$-reducts. This naturally requires a bridge result between the behaviour of a model and its standard expansion. The first and most basic is simply that for any $\CP$-formula $\phi(\bar{x})$ and $\bar{m}\in M\ ||\phi(\bar{m})||^M=||\phi(\bar{m})||^{M^\ast}$. For this reason we usually suppress reference to the standard expansion in the $||-||$ notation.
\end{rem}

Finally, we introduce an essential tool for model theory - back-\&-forth systems. These have been introduced for the many-valued setting in \cite{DellundeGarciaNoguera18:Fuzzybackandforth} for models defined over UL-algebras. In that discussion it is noted that the proofs can be easily generalised to apply to models over a wider choice of algebras and this includes residuated lattices. Importantly, whilst we introduced the basic definitions for arbitrary languages, the results for back-and-forth systems and our main investigation apply only to \textit{relational languages with object constants}, i.e. those without function symbols of positive arity. For this reason, from now on we assume all predicate  languages $\CP$ are relational with object constants. Given we will be working with a fixed defining algebra we make slight adjustments to the definitions and results presented in \cite{DellundeGarciaNoguera18:Fuzzybackandforth}. Specifically, rather than defining partial isomorphisms on models as a pair of functions $(f,g)$ where $f$ maps between the algebras and $g$ the underlying domains, we define them only as maps on the underlying domain, effectively assuming that $f$ is always the identity map. This allows us to replace the use of nesting depth, which is sensitive to the algebraic operators, with the more familiar quantifier depth introduced above. Nevertheless, the proof of the theorem is identical to \cite[Theorem 22]{DellundeGarciaNoguera18:Fuzzybackandforth}.

\begin{defi}
    Let $M,N$ be $\CP$-models. We say that a partial mapping $r\colon M\rightarrow N$ is a \textit{partial isomorphism} (p.iso) from $M$ to $N$ iff $r$ is an embedding on its domain, i.e. $r$ is injective and \begin{align*}
        &\text{for every }F\in\CP\ F^M\in dom(r) \text{ and }r(F^M)=F^N,\\
        &\text{for every }R\in\CP \text{ and } \bar{m}\in M:\bar{m}\in dom(r)\ R^M(\bar{m})=R^N(r(\bar{m})).
    \end{align*}
\noindent
    We say $M$ and $N$ are finitely isomorphic, denoted $M\isom_f N$ iff there is a sequence $\langle I_n:n\in\omega\rangle$ with the following properties:\begin{enumerate}
        \item Every $I_n$ is a non-empty set of partial isomorphisms from $M$ to $N$ and $I_{n+1}\subseteq I_n$.
        \item Forth: for every $r\in I_{n+1}$ and $m\in M\ \exists r'\in I_n:r\subseteq r'$ and $m\in dom(r')$.
        \item Back: for every $r\in I_{n+1}$ and $n\in N\ \exists r'\in I_n:r\subseteq r'$ and $n\in im(r')$. 
    \end{enumerate}
    We say that $M$ and $N$ are $k$-finitely isomorphic, denoted $M\isom_k N$ iff there is such a sequence $I_{j}:j\leq k$
\end{defi}
\begin{rem}
    Back-and-forth systems provide the second important bridge result between a model and its standard expansion, namely that $M\isom_k N$ iff $M^\ast\isom_k N^\ast$. This follows simply by the observation that any $\CP$ back-\&-forth system for $M$ and $N$ is immediately also a $\CP^\ast$ system from $M^\ast$ and $N^\ast$ and vice versa.
\end{rem}

\begin{thm}\label{thm:EF}
    Let $\CP$ be a finite relational language with object constants, $M,N$ be $\CP$-models and $k\in\omega$. If $M\isom_k N$ with $\bar{m}\mapsto\bar{n}\in I_k$ then $M,\bar{m}\equiv^s_k N,\bar{n}$.  
\end{thm}

\section{Hanf Locality}\label{sec:Hanf}

We now turn to our first locality result - Hanf locality. Hanf's theorem for classical logic is concerned with isomorphism types, a kind of syntactic encoding of b\&f-systems, applied to local neighbourhoods. Our first task is to set up a coherent many-valued counterpart to both these concepts. Throughout this section we define our models over a fixed well-connected residuated lattice $A$.

We start with the isomorphism types. An extension of the notion of diagrams, at a base level the idea is to define a formula that encodes partial isomorphisms and then extend these with formulas that encode the back and forth conditions. For our models, by making use of the standard expansions we can make an adjustment to the usual classical encoding as described in \cite{EbbinghausFlum95:Finmodeltheory}. 

\begin{defi}
    Let $\CP$ be a finite relational language with object constants and $A$ be a well-connected residuated lattice. Let $s\in\omega$ and $v_1,...,v_s$ be variables. For each finite $\CP$-model $M$, $m_1,...,m_s\in M$ and $k\in\omega$ we define the $\CP^\ast$-formula $\phi^k_{M,\bar{m}}(\bar{v})$ inductively as follows:\begin{align*}
        \phi^0_{M,\bar{m}}(\bar{v})&:= \bigland\{[R(w_1,...,w_t)\und a]\land [R(w_1,...,w_t)/a]\in \mathrm{Fm}:&\\
        & w_1,...,w_t\in\{v_1,...,v_s\}^t, R(w_1,...,w_t)\in \CP \text{ and }||R(\bar{m})||^M=a\}.&\\
        \phi^k_{M,\bar{m}}(\bar{v})&:=\bigland\limits_{m\in M}\exists v_{s+1}\phi^{k-1}_{\bar{m}m}(\bar{v},v_{s+1})\land\forall v_{s+1}\biglor\limits_{m\in M}\phi^{k-1}_{\bar{m}m}(\bar{v},v_{s+1}).&
    \end{align*}
    We call $\phi^k_{M,\bar{m}}(\bar{v})$ the $k$\textit{-isomorphism type} of $\bar{m}$ in $M$. We frequently drop reference to $M$ when it is clear from context.\footnote{The definition of $0$-isomorphism types is only sensitive to the behaviour of $\bar{m}$ in $M$ and independent from the rest of the model.}  
\end{defi}
\begin{rem}
    The set used in $\phi^0_{M,\bar{m}}(\bar{v})$ is finite for any selection of $\bar{m}$ and moreover as $M$ is finite by a simple induction on $k$ for $s,k\geq 0$ the set $\{\phi^k_{M,\bar{m}}:\bar{m}\in M^s$\} is finite. Thus the formulas are all well-defined. 

    Recall we have crisp equality in our language and so the $0$-isomorphism types contain in particular $(v_i=v_j)\und 1\land (v_i=v_j)/1$ iff $m_i=m_j$ and $v_i=v_j\und e\land v_i=v_j/e$ iff $m_i\not=m_j$ (where $e$ is the fixed value used for inequality in definition \ref{def:valuation}).  
\end{rem}

We can quickly establish some basic facts about isomorphism types. 
\begin{lem}\label{lem:isotype}
    For any $\CP$-model and $k\in\omega$ we have the following:
    \begin{enumerate}[label=\roman*.)]
        \item $qd(\phi^k_{\bar{m}}(\bar{v}))=k$,
        \item $M^\ast\models\phi^k_{\bar{m}}(\bar{m})$,
        \item For any $N$ and $\bar{n}\in N\ N^\ast\models\phi^0_{\bar{m}}(\bar{n})$ iff $\bar{m}\mapsto\bar{n}$ is a partial isomorphism (p.iso).
        \item For any $N$ and $\bar{n}\in N\ N^\ast\models\phi^k_{\bar{m}}(\bar{n})$ implies $\bar{m}\mapsto\bar{n}$ is a p.iso. 
    \end{enumerate}
\end{lem}
\begin{proof}
    $i.$ is immediate from the syntax construction and $ii.$ follows by an easy induction. For $iii.$ suppose $N^\ast\models\phi^0_{\bar{m}}(\bar{n})$. Firstly, the map $\bar{m}\mapsto\bar{n}$ is injective by our previous remark: \[m_i\not=m_j \Rightarrow v_i=v_j\und e\land v_i=v_j/ e\in \phi^0_{\bar{m}}(\bar{v}) \Rightarrow ||n_i=n_j||^{N^\ast}=e \Rightarrow n_i\not=n_j.\] 
    The only function symbols allowed are object constants $c\in\CP$ and: \[c^M=m_i\in\bar{m} \Rightarrow c=v_i\und 1\land c=v_i/1\in\phi^0_{\bar{m}}(\bar{v})\Rightarrow ||c=n_i||^N=1\Rightarrow c^N=n_i.\]
    Finally, letting $R\in\CP$ and $\bar{m'}\subseteq \bar{m}$ be such that $R^M(\bar{m'})=a$ then $R(\bar{v'})\und a\land R(\bar{v'})/a\in \phi^0_{\bar{m}}(\bar{v})$ which implies $N^\ast\models \phi^0_{\bar{m}}(\bar{n})\Rightarrow ||R(\bar{n'})\und a\land R(\bar{n'})/a||\geq 1$ and so $||R(\bar{n'})||^N=a=||R(\bar{m'})||^M$. Therefore $\bar{m}\mapsto\bar{n}$ is a p.iso. 

    For the converse, suppose $\bar{m}\mapsto\bar{n}$ is a p.iso then let $R(\bar{w})$ be any atomic $\CP$-formula with arity $t$ occurring in $\phi^0_{\bar{m}}(\bar{v})$. Then setting $a\coloneqq ||R(\bar{m})||^M$ we have $||R(\bar{n})||^N=||R(\bar{m})||^M=a$ and in particular $||R(\bar{n})\und a\land R(\bar{n})/a||^N\geq 1$, i.e. $N^\ast\models\phi^0_{\bar{m}}(\bar{n})$.

    For $iv.$ we proceed by induction on $k$ where $k=0$ is the first direction of $iii.$ Letting $k>0$ and supposing $N^\ast\models\phi^k_{\bar{m}}(\bar{n})$, in particular $N^\ast\models\exists v_{s+1}\phi^{k-1}_{\bar{m}m}(\bar{n}v_{s+1})$ for all $m\in M$. Choosing one arbitrarily, as $A$ is well-connected and $N$ finite $\exists n\in N:N^\ast\models\phi^{k-1}_{\bar{m}m}(\bar{n}n)$. By the induction hypothesis $\bar{m}m\mapsto\bar{n}n$ is a p.iso and so in particular so is $\bar{m}\mapsto\bar{n}$.
\end{proof}
\begin{rem}
    An immediate consequence of $iii.$ is given any two tuples $\bar{m}\in M$ and $\bar{n}\in N$ they realise the same type iff $\bar{m}\mapsto\bar{n}$ is a partial isomorphism.  
\end{rem}

These help prove their defining behaviour - their encoding of b\&f-systems. 

\begin{lem}\label{lem:typeEF}
    Let $\CP$ be a finite relational language with object constants and $A$ a well-connected residuated lattice. Let $M,N$ be finite models defined over $A$. Let $\bar{m}\in M$, $\bar{n}\in N$ and $k\in\omega$. The following are equivalent:\begin{enumerate}[nosep, label=\roman*.]
        \item $M\isom_k N$ with $\bar{m}\mapsto\bar{n}\in I_k$.
        \item $N^\ast\models\phi^k_{M,\bar{m}}(\bar{n})$.
        \item $M,\bar{m}\equiv^s_k N,\bar{n}$.
    \end{enumerate}
\end{lem}
\begin{proof}
    That $iii.$ implies $ii.$ is immediate from our previous lemma and noting that strong equivalence is preserved when moving to the standard expansion and theorem \ref{thm:EF} gives that $i.$ implies $iii.$. So suppose $N^\ast\models\phi^k_{\bar{m}}(\bar{n})$, we describe how to construct the b-\&-f system $\langle I_j\rangle_{j\leq k}$ recursively. 
    
    We define $I_k:=\{\bar{m}\mapsto\bar{n}\}$, this is a p.iso by part $iv.$ of the previous lemma and the assumption that $N^\ast\models\phi^k_{\bar{m}}(\bar{n})$. Indeed, for all $\bar{u}\mapsto\bar{v}\in I_k\ N^\ast\models\phi^k_{\bar{m}}(\bar{n})$. Suppose we have a constructed a set of p.isos $I_j$, with $0<j\leq k$ such that for all $\bar{u}\mapsto\bar{v}\in I_j\ N^\ast\models\phi^{j}_{\bar{u}}(\bar{v})$, and let $\bar{u}\mapsto\bar{v}\in I_j$. Now, for each $m\in M\ N^\ast\models\exists v_{s+1}\phi^{j-1}_{\bar{u}m}(\bar{v},v_{s+1})$ and so as $A$ is well-connected $\exists n\in N:N^\ast\models\phi^{k-1}_{\bar{u}m}(\bar{v},n)$. Again by $iv.$ of the previous lemma $\bar{u}m\mapsto\bar{v}n$ is a p.iso and so we may define the non-empty set of p.isos $I^{l}_{\bar{u},\bar{v}}:=\{\bar{u}m\mapsto\bar{v}n:N^\ast\models\phi^{j-1}_{\bar{u}m}(\bar{v},n)\}$. Similarly, letting $n\in N$ we have in particular that $N^\ast\models\biglor\limits_{m\in M}\phi^{j-1}_{\bar{u}m}(\bar{v},n)$ and as $A$ is well-connected there is an $m\in M:N^\ast\models\phi^{j-1}_{\bar{u}m}(\bar{v},n)$. Therefore we may once again define a non-empty set of p.isos $I^{r}_{\bar{u},\bar{v}}:=\{\bar{u}m\mapsto\bar{v}n:N^\ast\models\phi^{k-1}_{\bar{m}m}(\bar{n}n)\}$. We then define $I_{\bar{u},\bar{v}}=I^{l}_{\bar{u},\bar{v}}\cup I^{r}_{\bar{u},\bar{v}}$. 

    We may repeat this process on each element of $I_j$ for $j<k$ and take the total union to define $I_{j-1}$, observing that the result is a set of p.isos where $\forall \bar{u}\mapsto\bar{v}\in I_{j-1}\ N^\ast\models\phi^{j-1}_{\bar{u}}(\bar{v})$. Proceeding recursively, we define our sequence $\langle I_j\rangle_{j\leq k}$,  by construction each set is a non-empty set of p.isos, the sequence satisfies the b-\&-f conditions, and $\bar{m}\mapsto\bar{n}\in I_k$ as required. 
\end{proof}

Later we will wish to apply this result directly to standard expansions of models. Whilst the expanded language $\CP^\ast$ is not finite we can recover the equivalence with a little extra work.

\begin{cor}\label{cor:standardEF}
    Let $\CP$ be a finite relational language with object constants and $A$ a well-connected residuated lattice. Let $M,N$ be finite models defined over $A$. Let $\bar{m}\in M,\bar{n}\in N$ and $k\in\omega$. The following are equivalent:\begin{enumerate}[label=\roman*.]
        \item $M\isom_k N$ with $\bar{m}\mapsto\bar{n}\in I_k$.
        \item $M^\ast\isom_k N^\ast$ with $\bar{m}\mapsto\bar{n}\in I_k$.
        \item $N^\ast\models\phi^k_{M,\bar{m}}(\bar{n})$.
        \item $M,\bar{m}\equiv^s_k N,\bar{n}$.
        \item $M^\ast,\bar{m}\equiv^s_k N^\ast,\bar{n}$. 
    \end{enumerate}    
\end{cor}
\begin{proof}
    Our observation that any $\CP$ back-\&-forth system for $M$ and $N$ is also a $\CP^\ast$ system for $M^\ast$ and $N^\ast$ combined with the previous lemma gives the equivalence of $i.$ through $iv.$ and by taking the $\CP$ reduct $v.$ implies $iv.$ To conclude we show that $ii.$ implies $v.$. Let $\phi\in \mathrm{Fm}(\CP^\ast)$ be of quantifier depth at most $k$ and $\bar{a}$ be the finite list of the newly introduced nullary relation symbols occurring in $\phi$. We consider the finite relational language $\CP'=\CP\cup\{\bar{a}\}$ and the $\CP'$-reducts $M'$ and $N'$ of $M^\ast$ and $N^\ast$ respectively. Then $ii.$ implies that $M',\bar{m}\isom_k N',\bar{n}$ and we may apply our previous lemma to $M',\bar{m}$ and $N',\bar{n}$ to conclude $M',\bar{m}\equiv^s_k N',\bar{n}$. In particular $||\phi||^{M'}=||\phi||^{N'}$. Then as $\phi$ was arbitrary $M^\ast,\bar{m}\equiv^s_k N^\ast,\bar{n}$.  
\end{proof}

Our second concept for generalisation is distance in a model and the associated local neighbourhoods. Classically this is done via the Gaifman graph of a structure, the graph whose vertices are the elements of the model and which are adjacent iff they appear as part of a tuple in any true atomic relation. This lets us use the standard distance metric on graphs to define a distance metric on any classical model in an arbitrary signature. In our many-valued setting we can define a number of variants to the Gaifman graph metric; we start with a na\"{i}ve version where we again link the edges of the Gaifman graph to modelling. 
\begin{defi}\label{def:Gaifgraph}
    Let $A$ be a well-connected residuated lattice. We let $M$ be a $\CP$ model and define the Gaifman graph $G(M)$ as the graph with vertex set $M$ and $mEn$ iff there is an $R\in\CP$ and $\bar{s}\in M:m,n\in\bar{s}$ and $M\models R(\bar{s})$. We let $d(x,y)$ denote the standard distance metric on the graph $G(M)$, i.e.\ $d(m,n)$ is the length of the shortest path from $m$ to $n$ in $G(M)$. 
    
    For $\CP$ relational without constant symbols, for $r\in\omega$ and $m\in M$ we define the $r$-sphere of $m$ as:\[B(r,m):=\{n\in M:d(m,n)\leq r\},\] and use this interchangeably to denote the set of points and the induced \textit{strong} substructure of $M$, i.e.\ where each $R\in\CP$ is interpreted as the restriction of $R^M$ to $B(r,m)$. This naturally extends to use $B(r,m),m$ to denote the $\CP\cup\{c\}$-model where $c$ is interpreted as $m$. For $\bar{m}\in M^k$ we define $B(r,\bar{m}):=\bigcup\limits_{m\in\bar{m}} B(r,m)$.

    If $\CP$ does contain constant symbols we require that all $r$-spheres for $m$ contain the `common core' determined by the interpretation of the constant symbols.\begin{align*}
    B(r,\varnothing)&:=\bigcup\limits_{c\in\CP}\{n\in M:d(c^M,n)\leq r\},\\
    B(r,m)&:=\{n\in M:d(m,n)\leq r\}\cup B(r,\varnothing).
\end{align*}
\end{defi}

One can directly define the $r$-sphere type of a tuple $\bar{m}$ in model $M$ to be the $0$-isomorphism type of $B(r,\bar{m})$ - that is $\phi^0_{B(r,\bar{m})}(\bar{v})$. It then follows by Lemma \ref{lem:isotype} that for any other tuple $\bar{n}$ in a model $N\ N^\ast\models\phi^0_{B(r,\bar{m})}(B(r,\bar{n}))$ iff $B(r,\bar{m}),\bar{m}\isom B(r,\bar{n}),\bar{n}$ iff $M^\ast\models\phi^0_{B(r,\bar{n})}(B(r,\bar{m}))$. Note however that strictly speaking the formula $\phi^0_{B(r,\bar{m})}(\bar{v})$ depends on the exact listing of the elements of $B(r,\bar{m})$. This means we should really speak of \textit{an} $r$-sphere type of $\bar{m}$. Given the primary purpose of discussing $r$-sphere types is to have a shorthand way of saying that local neighbourhoods are isomorphic, it proves more convenient to shift perspective and consider isomorphism types directly via isomorphism equivalence classes of finite models.
\begin{defi}
    Let $A$ be a well-connected residuated lattice. By an \textit{isomorphism type} we mean an equivalence class for the isomorphism relation on $\CP\cup\{c_1,...,c_s\}$-models defined over $A$. We use the letters $\iota,\tau$ etc. to refer to isomorphism types. 
    
    If $\bar{m}\in M^s$ we say that $\bar{m}$ \textit{r}-\textit{realises} $\iota$, or equivalently that $\iota$ is \textit{the} $r$-sphere type of $\bar{m}$ iff $B(r,\bar{m}),\bar{m}$ belongs to $\iota$.  
\end{defi}
\begin{rem}\label{rem:isotypes}
    The direct via isomorphism definition aligns with syntactic definition of $0$-isomorphism types in the sense that two tuples $\bar{m}$ and $\bar{n}$ have the same $r$-sphere type iff the $r$-sphere of $\bar{n}$ realises a formula $\phi^0_{B(r,\bar{m})}(\bar{v})$ iff the $r$-sphere of $\bar{m}$ realises a formula $\phi^0_{B(r,\bar{n})}(\bar{v})$.  
\end{rem}

The application of isomorphism types to spheres is enough to state Hanf locality.
\begin{thmC}[\cite{Libkin04:Elementsfinmod}] \emph{Hanf Locality (Classical)}
    Let $M,N$ be finite (classical) $\CP$-models and $k\in\omega$. Let $\langle r_j\rangle_{j\leq k}$ be recursively defined by $r_0=0$, $r_{j+1}=3r_j+1$. Let $e$ be the maximum size of the $r_k$-spheres in $M$ and $N$ and suppose that for each $r\leq r_k$ and sphere type $\iota$ either $i.$ or $ii.$ holds:\begin{enumerate}[label=\roman*.]
        \item $M$ and $N$ have the same number of elements that $r$-realise $\iota$.
        \item Both $M$ and $N$ have more than $k\cdot e$ elements that $r$-realise $\iota$. 
    \end{enumerate}
    Then $M\equiv_k N$.
\end{thmC}
\noindent
This theorem fails in any residuated lattices with 2 distinct elements not above the unit element, including all the standard algebras discussed in examples \ref{exa:wellconnectboundresi} and \ref{exa:resichaincoatom}.
\begin{exa}
    Let $\CP$ be the language of weighted graphs, i.e. consist of a single binary relation symbol $E$. Let $A$ be any well-connected residuated lattice such that $\exists a,b\in A:a,b\not\geq 1$ and $a\not=b$. We define the models $M$ and $N$ with the same two element domain $\{s,t\}$ with interpretation of $E$:\begin{align*}
        E^M(s,s)=1 && E^M(t,t)=1 && E^N(s,s)=1 && E^N(t,t)=1\\
        E^M(s,t)=a && E^M(t,s)=a && E^N(s,t)=b && E^N(t,s)=a.
    \end{align*}

    Consider the sentence $\phi:=\forall x\forall y(xEy\und yEx)$. One can easily verify that $||\phi||^M=1$ and $||\phi||^N=(a\und b)\land (b\und a)<1$ (this is because $a\not= b$ and $x\und y\geq 1$ iff $x\leq y$ by the residuation property). Therefore $M\not\equiv_2 N$. 

    By contrast $M$ and $N$ satisfy the assumption in the Hanf theorem. For $r\in\omega$ let $B_M(r,s)$ be the $r$-sphere of $s$ in $M$ and $B_N(r,s)$ be the $r$-sphere of $s$ in $N$, similarly for $t$. Now, $B_M(r,s)=\{s\}=B_N(r,s)$ and $B_M(r,t)=\{t\}=B_N(r,t)$, and moreover the induced strong substructures are isomorphic given that $E^M$ and $E^N$ agree on $(s,s)$ and $(t,t)$. This means the $r$-sphere types of $s$ in $M$ and $N$ and $t$ in $M$ and $N$ are the same, and so for any sphere type $\iota\ M$ and $N$ have the same number of elements that $r$-realise $\iota$.
\end{exa}

There are natural alternatives to the modelling based Gaifman distance metric and this motivates the question of whether we might recover Hanf locality for some appropriate metric.
\begin{defi}
    Let $A$ be a well-connected residuated lattice and $t\in A$. For a $\CP$-model we define its $t_\geq$-Gaifman graph as the graph with vertex set $M$ and $mEn$ iff there is an $R\in\CP$ and $\bar{s}\in M:m,n\in\bar{s}$ and $||R(\bar{s})||^M\geq t$. We let $d_{\geq t}(x,y)$ denote the standard distance metric on the $t_\geq$-Gaifman graph of $M$. We define the corresponding $t$-sphere $B_t(r,\bar{m})$ and $t,r$-sphere type of $\bar{m}$ analogously.

    We can repeat this again for the $t_>$-graph and distance metric defined $d_{>t}(x,y)$ with $>$ in place of $\geq$, denoted the resulting strict $t$-spheres by $B_{t'}(r,\bar{m})$. We refer to these distance metrics as the (strict) threshold distance metrics. When the specific distance metric $d_{\geq t}$ or $d_\{>t\}$ is clear we drop reference to $t$ and simply write $d(x,y)$ and $B(r,\bar{m})$. 
\end{defi}

Minor adjustments to the previous example provide counterexamples to a potential Hanf theorem based on almost all these distance metrics whenever there are distinct elements beneath the threshold. The failure in each case is simply that distance in neighbourhoods only contains information about the values of $R^M$ above the threshold. The exception is therefore when $A$ is bounded and we use the strict distance metric based on the lower bound - $d_{>\bot}(x,y)$.\footnote{This is the version of distance employed in both \cite{BiziereGradelNaaf23:Localitysemiring} and \cite{WildSchroderPattinsonKonig18:fuzzyvanBenthem}. It coincides with the modelling based metric for well-connected bounded residuated lattice where $b\geq 1$ iff $b>\bot$.} In this case elements being non-adjacent in the Gaifman graph exactly specifies all $R^M$ values involving said elements - they must take value $\bot$. Accordingly, we can recover a Hanf theorem result. Really what this behaviour allows is the recovery of a technical lemma establishing that under certain conditions we can combine the isomorphisms of local neighbourhoods based on the strict $\bot$ metric. 
\begin{lem}
    Let $M,N$ be $\CP$-models over a well-connected bounded residuated lattice $A$. Let $r\in\omega, \bar{m}\cup\{m\}\subseteq M$ and $\bar{n}\cup\{n\}\subseteq N$. Suppose that:\begin{enumerate}[label=\roman*.]
        \item $m\not\in B(2r+1,\bar{m})$,
        \item $n\not\in B(2r+1,\bar{n})$,
        \item $B(r,\bar{m})\isom B(r,\bar{n})$
        \item $B(r,m)\isom B(r,n)$.
    \end{enumerate}
    Then $B(r,\bar{m}m)\isom B(r,\bar{n}n)$. 
\end{lem}
\begin{proof}
    For the moment suppose that $\CP$ contains no constant symbols. Letting $\pi$ and $\pi'$ be the respective isomorphisms we consider the map $\pi^+:=\pi\cup\pi'$. Conditions $i.$ and $ii.$ imply that $B(r,\bar{m})\cap B(r,m)=\varnothing=B(r,\bar{n})\cap B(r,n)$ so this is a well-defined function and a bijection, it remains to check it is an isomorphism. 
    
    We first note that for any $s\in B(r,\bar{m})$ and $t\in B(r,m)$, if $d_{>\bot}(s,t)\leq 1$ then $d_{>\bot}(\bar{m},m)\leq 2r+1$ and $m\in B(2r+1,\bar{m})$ which is a contradiction. Therefore $d_{>\bot}(s,t)>1$ and in particular for any $s,t\in B(r,\bar{m}m)$ if $d_{>\bot}(s,t)\leq 1$ then either $s,t\in B(r,\bar{m})$ or $s,t\in B(r,m)$. Similarly, for any $s,t\in B(r,\bar{n}n)$ if $d_{>\bot}(s,t)\leq 1$ then either $s,t\in B(r,\bar{n})$ or $s,t\in B(r,n)$.
    
    So now let $R\in\CP$ and $\bar{s}\in B(r,\bar{m}m)$. We have three cases: $R^M(\bar{s})>\bot$, $R^N(\pi^+(\bar{s}))>\bot$ or both $R^M(\bar{s})=\bot=R^N(\pi^+(\bar{s}))$. Our check for the final case is immediate, so suppose $R^M(\bar{s})>\bot$. Then for all $i,j\ d_{>\bot}(s_i,s_j)\leq 1$ and so either $\bar{s}\in B(r,\bar{m})$, and then \[R^{B(r,\bar{m}m)}(\bar{s})=R^{B(r,\bar{m})}(\bar{s})=R^{B(r,\bar{n})}(\pi(\bar{s}))=R^{B(r,\bar{n}n)}(\pi^+(\bar{s})),\] or $\bar{s}\in B(r,m)$ and similarly \[R^{B(r,\bar{m}m)}(\bar{s})=R^{B(r,m)}(\bar{s})=R^{B(r,n)}(\pi'(\bar{s}))=R^{B(r,\bar{n}n)}(\pi^+(\bar{s})).\]
    If $R^N(\pi^+(\bar{s}))>0$ then for all $i,j\ d_{>\bot}(\pi^+(s_i),\pi^+(s_j))\leq 1$ and either $\pi^+(\bar{s})\in B(r,\bar{n})$ or $\pi^+(\bar{s})\in B(r,n)$ and we can make the same identities. 

    Now we consider the case where we allow constant symbols. In this case, $B(r,\bar{m})\cap B(r,m)=B(r,\varnothing)$. The potential problem with combining our isomorphisms is they may conflict on the common core. However, the assumption that $m\not\in B(2r+1,\bar{m})$ further implies that $\{t\in M:d_{>\bot}(t,m)\leq r\}\cap B(r,\varnothing)=\varnothing$ and thus we may 'break' in favour of the isomorphism for $B(r,\bar{m})$. More precisely, we let $\pi$ be the isomorphism between $B(r,\bar{m})$ and $B(r,\bar{n})$, $\pi'$ be the restriction of the isomorphism between $B(r,m)$ and $B(r,n)$ onto the sets $\{t\in M:d_{>\bot}(t,m)\leq r\}$ and $\{t\in N:d_{>\bot}(t,n)\leq r\}$ and consider the map $\pi^+=\pi\cup\pi'$. This is a well-defined bijection due to acting on disjoint sets and we again must check it is an isomorphism. For constant symbols this is immediate as $\pi^+(c^{B(r,\bar{m}m)})=\pi(c^{B(r,\bar{m})})=c^{B(r,\bar{n})}=c^{B(r,\bar{n}n)}$, and for relation symbols this is just as before except we work with the pairs of disjoint sets $B(r,\bar{m})$, $\{t\in M:d_{>\bot}(t,m)\leq r\}$ and $B(r,\bar{n})$, $\{t\in N:d_{>\bot}(t,n)\leq r\}$. 
\end{proof}

\begin{thm}[Hanf's Theorem for Residuated Lattices]\label{thm:Hanf}
    Let $M,N$ be finite $\CP$-models over a well-connected residuated bounded lattice $A$, $d$ be the strict $\bot$-threshold distance and $k\in\omega$. We define the sequence $\langle r_j\rangle_{j\leq k}$ by $r_j=(3^j-1)/2$. Let $e$ be the maximum size of the $r_k$-spheres of each element in $M$ and $N$ and suppose that for each $r\leq r_k$ and isomorphism type $\iota$ either $i.$ or $ii.$ holds:\begin{enumerate}[label=\roman*.]
        \item $M$ and $N$ have the same number of elements that $r$-realise $\iota$.
        \item Both $M$ and $N$ have more than $k\cdot e$ elements that $r$-realise $\iota$. 
    \end{enumerate}
    Then $M\equiv^s_k N$.
\end{thm}
\begin{proof} 
    We are able to essentially follow the textbook proof for the classical case (e.g. Ebbinghaus and Flum \cite[Theorem 1.4.1]{EbbinghausFlum95:Finmodeltheory} or \cite[Theorem 4.24]{Libkin04:Elementsfinmod}). We show that $\langle I_j\rangle_{j\leq k}$ is a b\&f-system from $M$ and $N$ where we define:\[I_j:=\{\bar{m}\mapsto\bar{n}\text{ p.iso for }M,N:(B(r_j,\bar{m}),\bar{m})\isom (B(r_j,\bar{n}),\bar{n})\text{ and }|\bar{m}|\leq k-j\}.\]
    
    First, observe that we can equivalently define the sequence $r_j$ by $r_j=0$ and $r_{j+1}=3r_j+1$. For $j=k$ we have $\bar{m}=\varnothing$ and take $I_k=\{\varnothing\mapsto\varnothing\}$ as a trivial partial isomorphism. This still requires checking that $B(r_k,\varnothing)$ as defined in $M$ is isomorphic to $B(r_k,\varnothing)$ as defined in $N$. This follows easily by our assumption, if we let any $m\in M$ and $\iota$ be its $r_k$-sphere type, $\exists n\in N$ with $r_k$-sphere type $\iota$ and so $B(r_k,m)\isom B(r_k,n)$. These spheres contain the local neighbourhoods of the interpreted constant symbols in their respective models, so we can restrict our isomorphism onto those neighbourhoods.     

    Now we turn to the back and forth conditions. By symmetry it is enough to prove the forth condition. Accordingly, let $0\leq j<k,m\in M$ and $\bar{m}\mapsto\bar{n}\in I_{j+1}$ witnessed by say $\pi$. 

    Case $1$: $m\in B(2r_j+1,\bar{m})$. In this case $B(r_j,\bar{m}m)\subseteq B(r_{j+1},\bar{m})$, as for any $t\in B(r_j,\bar{m}m)$ either $t\in B(r_j,\bar{m})\subseteq B(r_{j+1},\bar{m})$ or $d(t,m)\leq r_j$ and then there is some $x\in \bar{m}\cup\{c^M:c\in\CP\}:d(t,x)\leq d(t,m)+d(m,x)\leq r_j+2r_j+1=3r_j+1=r_{j+1}$). Therefore we can consider $n=\pi(m)$ and $\pi$ restricts to an isomorphism on $B(r_j,\bar{m}m)\isom B(r_j,\bar{n}n)$. Then $\bar{m}m\mapsto\bar{n}n\in I_j$ as required. 

    Case $2$: $m\not\in B(2r_j+1,\bar{m})$. We let $\iota$ be the $r_j$-sphere type of $m$. By the existence of $\pi$ we claim that $B(2r_j+1,\bar{m})$ and $B(2r_j+1,\bar{n})$ have the same number of elements of $r_j$-sphere type $\iota$. For any $s\in B(2r_j+1,\bar{m})\ B(r_j,t)\subseteq B(3r_j+1,\bar{m})=B(r_{j+1},\bar{m})$ and therefore $B(r_j,s),s\isom \pi[B(r_j,s),s]=B(r_j,\pi(s)),\pi(s)$. Thus $s$ and $\pi(s)$ have the same $r_j$-sphere type and the injectivity of $\pi$ implies the number of elements of $r_j$-sphere type $\iota$ in $B(2r_j+1,\bar{m})$ is less than or equal to that number in $B(2r_j+1,\bar{n})$. By symmetry with $\pi^{-1}$ we have equality. Now, if $i$. holds for $\iota$ as $m\in M\setminus B(2r_j+1,\bar{m})$ and is of $r_j$-sphere type $\iota$ there must exist an $n\in N\setminus B(2r_j+1,\bar{n})$ of $r_j$-sphere type $\iota$. 
    
    In the case that $ii.$ holds for $\iota$, we note that the size of $B(2r_j+1,m_i)$ is $\leq e$ by assumption and so $|B(2r_j+1,m)|\leq$ length$(\bar{m})\cdot e\leq k\cdot e$. In particular, the number of elements of $r_j$-sphere type $\iota$ in $B(2r_j+1,\bar{m})$ and $B(2r_j+1,\bar{n})$ is at most $k\cdot e$ and therefore $\exists n\in N\setminus B(2r_j+1,\bar{n})$ of $r_j$-sphere type $\iota$.  
    
    In either case we have found $n\not\in B(2r_j+1,\bar{n})$ with $B(r_j,m)\isom B(r_j,n)$ and we may apply our previous lemma to conclude that $B(r_j,\bar{m}m),\bar{m}m\isom B(r_j,\bar{n}n),\bar{n}n$ witnessing that $\bar{m}m\mapsto\bar{n}n\in I_j$.
\end{proof}

The fact we can follow the classical proof of Hanf's theorem essentially unmodified emphasises that the real interest is in the condition within the theorem. We will return to this idea in more detail in Section~\ref{sec:queries}. For now, as an immediate first step we note that it implies another form of locality, due to Fagin, Stockmeyer and Vardi \cite{FaginStockmeyerVardi94:monadicNPcoNP}.
\begin{defi}
    Let $A$ be a well-connected residuated lattice. Let $M,N$ be finite $\CP$-models over $A$. We write $M\leftrightarrows_{r} N$ iff for each isomorphism type $\tau$ of $\CP\cup\{c\}$-models the number of elements of $M$ and $N$ with $r$-sphere type $\tau$ is the same. 
\end{defi}
\begin{rem}
    By remark \ref{rem:isotypes} this condition is equivalent to there being a bijection $f\colon M\rightarrow N$ such that for all $m\in M\ B(r,\bar{m}),\bar{m}\isom B(r,f(\bar{m})),f(\bar{m})$.
\end{rem}
 
\begin{cor}\label{cor:Hanf}
    Let $A$ be a well-connected residuated lattice. Let $M,N$ be finite $\CP$-models over $A$ and $k\in\omega$. Suppose $M,\bar{m}\leftrightarrows_{3^k} N,\bar{n}$. Then $M,\bar{m}\equiv^s_k N,\bar{n}$. 
\end{cor}
\begin{proof}
    We aim to show $M,\bar{m}$ and $N,\bar{n}$ satisfy the Hanf theorem criterion. By assumption, for each isomorphism type $\tau\ M$ and $N$ have the same number of elements that $3^k$-sphere realise $\tau$. Moreover $r_k=(3^k-1/2)\leq 3^k$, so the corollary follows upon checking that if $M,\bar{m}$ and $N,\bar{n}$ have the same number of elements of $r$-sphere type $\tau$ then they also have the same number of elements of $l$-sphere type $\iota$ for all $l\leq r$. But this is clear; two elements have the same $r$-sphere type iff $B(r,x)\isom B(r,y)$ and have the same $l$-sphere type iff $B(l,x)\isom B(l,y)$ and two spheres being isomorphic certainly implies smaller spheres are.  
\end{proof}
\begin{rem}
    We use the lax bound of $3^k$ for cleanness, we could of course state this corollary using the tighter bound for Hanf locality of $(3^k-1)/2$ in line with theorem \ref{thm:Hanf}.
\end{rem}


\section{Gaifman Locality}\label{sec:Gaif}
We now want to follow the pattern of Hanf locality with Gaifman locality; set up the many-valued counterparts to the key concepts of the proof and demonstrate the resulting theorem goes through as it does classically. There are however some important differences. Firstly, we will need to recover some additional syntactic machinery which puts us into a more restricted setting than our quite general one for Hanf. We will introduce these restrictions alongside the machinery they are required for. Secondly, we do not recover the usual classical Gaifman theorem. In the classical setting, a common proof of Gaifman's theorem (for both finite and infinite models) proceeds via a main lemma to which the theorem comes as a corollary \cite[Theorem 1.5.1]{EbbinghausFlum95:Finmodeltheory}. This implication relies on compactness. 
\begin{propC}[\cite{EbbinghausFlum95:Finmodeltheory}]
    Let $\Phi$ be a set of first-order classical sentences. Assume that any two structures that satisfy the same sentences of $\Phi$ are elementary equivalent. Then any first-order sentence is equivalent to a boolean combination of sentences of $\Phi$.
\end{propC}
\begin{proof}[Proof sketch]
    Take an arbitrary (classical) structure $A$ and consider \[\Phi(A):=\{\psi\in\Phi:A\models\psi\}\cup\{\neg\psi:\psi\in\Phi, A\models\neg\psi\}.\]
    It follows from the assumption on $\Phi$ that for any model $B$, $B\models\phi$ iff there is some model $A$ such that $A\models\phi$ and $B\models\Phi(A)$. In particular, for each $A\models\phi\ \Phi(A)\models\phi$, applying compactness there is a finite set of sentences $\Phi_0(A)$ such that $\Phi_0(A)\models\phi$. Moreover, for \emph{any} finite collection $\{A_i:A_i\models\phi\}$ \[\biglor\limits_{A_i}\bigland\Phi_0(A_i)\models\phi.\]
    Suppose (for contradiction) that there is no finite collection $\{A_i:A_i\models\phi\}$ such that the reverse holds. Then another application of compactness yields a model $B$ such that $B\models\phi$ but for each $A$ such that $A\models\phi\ B\not\models\Phi(A)$, contradicting our equivalence. Therefore this is some finite collection $\{A_i\}$ for which the reverse holds and $\phi$ is equivalent to the sentence: \[\psi\coloneqq \biglor_{A_i}\bigland\Phi_0(A_i).\qedhere\]
\end{proof}

The status of compactness for non-classical models is a significant area of research in its own right. As a brief mention of some positive results, compactness has been recovered for models defined over a fixed finite MTL-algebra \cite[Theorem 4.4]{Dellunde13:FuzzyUltra} and models defined over the standard $[0,1]$ based \L ukasiewicz and G\"{o}del algebras, but fails for the $[0,1]$ product algebra \cite{Hajek02:Technical}. Here our restriction to finite models proves significant. Just as compactness fails for finite classical models \cite[Chapter 0]{EbbinghausFlum95:Finmodeltheory} the same formulas witness its failure for finite models defined over well-connected residuated lattices. 

The main lemma in the proof of Gaifman locality, which we refer to as Gaifman's lemma, is a proof that there is a set of sentences defined using only local information that satisfies the hypothesis in the preceding proposition. It is this lemma that we generalise into the abstract setting. This is still of significant interest, as we will see later it is enough to recover the application to definable queries. 
\begin{thm}[Gaifman's Lemma (Classical)]
    Suppose $A$ and $B$ satisfy the same basic local sentences. Then $A\equiv B$.
\end{thm}
\noindent
Basic local sentences are a particular syntactic form based on local formula and it is this we need to form an abstracted version of. This is enabled by another syntactic encoding, this time of the distance metric. We define the formula $\theta_r(x,y)$ inductively ($a(R)$ is the arity of the relation symbol $R\in\CP$):
\begin{align*}
    \theta_0(x,y)&:=x=y \\
    \theta_1(x,y)&:=\biglor\limits_{R\in\CP}\exists u_1...\exists u_{a(R)}[R(u_1,...,u_{a(R)})\land\biglor\limits_{1\leq i,j\leq a(R)}(u_i=x\land u_j=y)] \\
    \theta_{r+1}(x,y)&:=\exists z(d_k(x,z)\land d_1(z,y)).
\end{align*}

When applied to finite models over a well-connected residuated lattice $A$ these formulas capture a refined measure of distance. For any $\CP$-model over $A$ we define the Gaifman matrix $G_v(M)$ as the $|M|\times |M|$ matrix with components: \[v_{m,n}:=\sup\{R^M(\bar{c})\in A:R\in \CP, m,n\in\bar{c}\},\] that is the supremum of values taken by atomic formula in which $m$ and $n$ occur.

A path through $G_v(M)$ is a sequence of elements $M$ and for a path of length $r-1$, its label is the $r$-sequence $\langle v_{m_i,n_i}\rangle_{i\leq r}$ and its weight is $\inf\{v_{m_i,n_i}\in A:1\leq i\leq r\}$. It is straightforward to verify that $||\theta_r(m,n)||^M$ is the supremum of the weight of paths through $G_v(M)$ from $m$ to $n$ of length at most $r$ and as corollary to this that it encodes the distance in the na\"{i}ve Gaifman graph. 

\begin{cor}
    Let $A$ be a well-connected residuated lattice, $M$ a finite $\CP$-model over $A$, $m,n\in M$ and $r\in\omega$. \[M\models \theta_r(m,n)\text{ iff }||\theta_r(m,n)||^M\geq 1\text{ iff }d(m,n)\leq r\text{ iff } n\in B(r,m).\] 
\end{cor}
\begin{rem}
    Note that we only ask for our distance formulas to encode distance at the level of the modelling relation.\footnote{This contrasts with the approach in \cite{BiziereGradelNaaf23:Localitysemiring} which requires an encoding of distance that allows for quantification within distance to take the actual value. This difference is driven by the needs of the respective investigations which we will return to in Section~\ref{sec:semirings}.}
\end{rem}

This corollary relies on our limited form of witnessing applying to both the quantifiers $\exists/\forall$ and the lattice connectives $\lor/\land$. Given any $t\in A$, if this limited form of witnessing holds at $t$ then we can encode the threshold distance metric $d_t(x,y)$ for the $t$-Gaifman graph using the standard expansion of a model. In particular, when $M$ is finite and $A$ is a chain this applies to all $t\in A$. 

\begin{cor}
    Let $A$ be a residuated chain, $t\in A$, $M$ a finite $\CP$-model over $A$, $m,n\in M$ and $r\in\omega$. \[M^\ast\models t\und \theta_r(m,n)\text{ iff }||\theta_r(m,n)||^M\geq t\text{ iff }d_{\geq t}(m,n)\leq r\text{ iff }n\in B_t(r,m).\]
\end{cor}
\noindent
The classical distance formulas naturally come paired with their negations, with the property that $M\models\neg \theta_r(m,n)$ iff $d(m,n)>M$, and these are essential to the definition of local formulas/basic local sentences. In general residuated lattices do not have such a classical-like negation term, but we can consider a still interesting smaller class which do and consequently for which `far' distance formulas can be defined, namely those $A$ with a \textit{co-atom}. Recall an element $c\in A$ is a co-atom iff $\forall a\in A\ a\leq c$ iff $1\not\leq a$. The co-atom lets us define a classical like negation term $\tau(x):= x\und c$ in the language $P^A$ where $M^\ast\models\phi$ iff $M^\ast\not\models t(\phi)$. This lets us define the distance formulas as we did classically namely: \[\tau(t\und\theta_r(x,y)).\]
In the case of the strict $t$-threshold distance metrics the situation is reversed, the standard expansion of a model can directly encode elements being far from each other and we use the negation like term to encode elements being close. 
\begin{align*}
    M^\ast\models\theta_r(m,n)\und t\text{ iff }||\theta_r(m,n)||^M\leq t&\text{ iff }d_{>t}(m,n)> r\text{ iff }n\not\in B_{t'}(r,m).\\
    M^\ast\models (\theta_r(m,n)\und t)\und c\text{ iff }||\theta_r(m,n)||^M>t&\text{ iff }d_{>t}(m,n)\leq r\text{ iff }n\in B_{t'}(r,m).
\end{align*}
It is worth noting that for a given (strict) $t$-threshold distance we could equally use an element $c$ that acts like a (atom) co-atom for $t$ itself, i.e. $\forall a\in A\ t\not\leq a$ iff $a\leq c$ or $\forall a\in A\ t\not< a$ iff $c\leq a$ respectively. We could then directly encode being far by $\theta_r(m,n)\und c$ or being near by $c\und \theta_r(m,n)$ respectively. What follows could be equally completed using any of the encodings for distance described, should that be appropriate for a particular algebra and distance metric combination. For ease of presentation we will work with the assumption that our algebra has a $1$ co-atom. We proceed working with an arbitrary (strict) threshold distance metric and use $\delta_r(x,y)$ and $\rho_r(x,y)$ as shorthand for any appropriate encoding of distance, i.e. formulas such that\[M^\ast\models\delta_r(m,n)\text{ iff }d(m,n)\leq r,\, M^\ast\models\rho_r(m,n)\text{ iff }d(m,n)>r.\]

Classically, the encoding of distance allows one to associate to any formula $\phi(\bar{x},\bar{y})$ a formula $\phi^{r,\bar{x}}(\bar{x},\bar{y})$ such that for any classical model $M$, $\bar{m}\in M$ and $\bar{n}\in B(r,\bar{m})$: \[M\models\phi^{r,\bar{x}}(\bar{m},\bar{n}) \text{ iff } B(r,\bar{m})\models\phi(\bar{m},\bar{n}).\] 
This association is the basis of the definition of local formula. In the many-valued setting, because our distance formulas take some actual value in the algebra we have to be more careful when trying to construct a similar association. This prompts our second restriction to witnessed $\CP$-models. In witnessed models we can recover half the proof for all the quantifier shifts for $\exists$ and $\forall$. The other half of the proof follows when we note that for residuated lattices every connective $\land,\lor,\cdot,\und,/$ is either monotone (order-preserving) or antitone (order inverting) in each position. 
\begin{prop}  
    Let $\CL$ be an algebraic signature and $\CP$ a predicate language. Let $A$ be a $\CL$-algebra and $M$ a $\CP$-model over $A$. Let $\phi,\lambda_1,...,\lambda_n$ be $\CP$-formulas where $x$ is not free in any $\lambda_i$ and $\circ\in\CL$ be any $n+1$-ary connective.

    \noindent
    If $M$ is a witnessed $\CP$-model then:\begin{enumerate}[label=\roman*.]
    \item $||\forall x(\circ(\phi,\lambda_1,...,\lambda_n))||^M\leq ||\circ(\exists x\phi, \lambda_1,...,\lambda_n)||^M\leq||\exists x(\circ(\phi,\lambda_1,...,\lambda_n))||^M$,
    \item $||\forall x(\circ(\phi,\lambda_1,...,\lambda_n))||^M\leq||\circ(\forall x\phi,\lambda_1,...,\lambda_n)||^M\leq ||\exists x(\circ(\phi,\lambda_1,...,\lambda_n))||^M$,
    \end{enumerate}
    \noindent
    If $\circ$ is monotone for the position of $\phi$ we have: \begin{enumerate}
        \item [iii.] $||\exists x(\circ(\phi,\lambda_1,...,\lambda_n))||^M\leq||\circ(\exists x\phi,\lambda_1,...,\lambda_n)||^M$,
        \item [iv.] $||\circ(\forall x\phi, \lambda_1,...,\lambda_n)||^M\leq||\forall x(\circ(\phi,\lambda_1,...,\lambda_n))||^M$,
    \end{enumerate}
    \noindent
    When $\circ$ is antitone for the position of $\phi$ we have:\begin{enumerate}
        \item [v.] $||\exists x(\circ(\phi,\lambda_1,...,\lambda_n))||^M \leq ||\circ(\forall x\phi,\lambda_1,...,\lambda_n)||^M$,
        \item [vi.] $||\circ(\exists x\phi,\lambda_1,...,\lambda_n)||^M\leq||\forall x(\circ(\phi,\lambda_1,...,\lambda_n))||^M$. 
    \end{enumerate}
    \noindent
    Consequently, if every position of each connective $\circ\in\CL$ is monotone or antitone then for every formula $\phi$ there is a formula $\psi$ in \textit{prenex normal form}, i.e. a string of quantifiers applied to a quantifier-free formula that is strongly equivalent to $\phi$, i.e. such that for any finite $\CP$ model $M$ and $\bar{m}\in\CP$ \[||\phi(\bar{m})||^M=||\psi(\bar{m})||^M\]
\end{prop}
\begin{proof}
    i. By $\exists$-witnessing \[||\circ(\exists x\phi,\lambda_1,...,\lambda_n)||=\circ^A(||\exists x\phi||,||\lambda_1||,...,||\lambda_n||)=\circ^A(||\phi(m)||,||\lambda_1||,...,||\lambda_n||)\] 
    for some $m\in M$. Therefore:\begin{multline*}
    ||\forall x(\circ(\phi,\lambda_1,...,\lambda_n)||=\inf_{m}\{\circ^A(||\phi(m)||,||\lambda_1||,...,||\lambda_n||)\}\\
    \leq ||\circ(\exists x\phi,\lambda_1,...,\lambda_n)||\leq \sup_{m}\{\circ(||\phi(m),\lambda_1,...,\lambda_n)||)\}=||\exists x(\circ(\phi,\lambda_1,...,\lambda_n))||
    \end{multline*} 
    The proof for ii. is symmetric. 

    iii. Letting $m\in M$, $\circ$ monotone for $\phi$'s position implies:\begin{multline*}
    ||\circ(\phi(m),\lambda_1,...,\lambda_n)||=\circ^A(||\phi(m)||,||\lambda_1||,...,||\lambda_n||)\\\leq \circ^A(||\exists x\phi||,||\lambda_1||,...,||\lambda_n||)=||\circ(\exists x\phi,\lambda_1,..,\lambda_n)||.
    \end{multline*}
    Again iv. is symmetric. 

    v. Let $m\in M$, as $\circ$ is antitone for $\phi$'s position, \[\circ^A(||\phi(m)||,||\lambda_1||,...,||\lambda_n||)\leq \circ^A(||\forall x\phi||,||\lambda_1||,...,||\lambda_n||)=||\circ(\forall x\phi,\lambda_1,...,\lambda_n)||.\]
    As $m\in M$ was arbitrary \[||\exists x(\circ(\phi,\lambda_1,...,\lambda_n))||\leq ||\circ(\forall x\phi,\lambda_1,...,\lambda_n)||.\] 
    Again vi. is symmetric. 

    The conclusion for prenex normal form is then a straightforward induction utilising each of the quantifier shifts and renaming of variables as appropriate.
\end{proof}

We define our relative quantified formulas for formulas in prenex normal form, and factor through these for formulas not in prenex normal form.
\begin{defi}
    Let $\phi(\bar{x},\bar{y})$ be a formula in prenex normal form. We define its relativisation $\phi^{r,\bar{x}}(\bar{x},\bar{y})$ to $B(r,\bar{m})$ using the distance formulas by elimination of the outermost quantifiers.  
    \begin{align*}
        [\exists z\psi(\bar{x},\bar{y},z)]^{r,\bar{x}}&:= \exists z(\delta_r(\bar{x},z)\land \psi^{r,\bar{x}}(\bar{x},\bar{y},z)). \\
        [\forall z\psi(\bar{x},\bar{y},z)]^{r,\bar{x}}&:= \forall z(\rho_r(\bar{x},z)\lor \psi^{r,\bar{x}}(\bar{x},\bar{y},z)).
    \end{align*}
    Letting $\phi(\bar{x},\bar{y})$ be any $\CP$-formula, and $Q_1z_1...Q_nz_n\psi(\bar{x},\bar{y},\bar{z})$ be a strongly equivalent formula in prenex normal form (with $Q_i\in\{\exists,\forall\}$). We define the relativisation of $\phi$ via $\psi$.
    \[\phi^{r,\bar{x}}(\bar{x},\bar{y})=[Q_1z_1...Q_nz_n\psi(\bar{x},\bar{y},\bar{z})]^{r,\bar{x}}\] 
\end{defi}

\begin{lem}\label{lem:localform}
    Let $A$ be a residuated well-connected lattice with co-atom and consider an encodable distance metric for finite $\CP$-models over $A$. Letting $M$ be a finite witnessed $\CP$-model over $A$, for any $\CP^\ast$-formula $\phi(\bar{x},\bar{y})$ any $\bar{m}\in M$, $r\in\omega$, $\bar{n}\in B(r,\bar{m})$ we have:
    \[M^\ast\models\phi^{r,\bar{m}}(\bar{m},\bar{n})\text{ iff }B(r,\bar{m})^\ast\models\phi(\bar{m},\bar{n}).\]
\end{lem}
\begin{proof}
    It is enough to prove the lemma for $\phi$ in prenex normal form; we proceed by induction on the number of quantifiers in $\phi$. When $\phi$ is quantifier free $\phi^{r,\bar{x}}(\bar{x},\bar{y})=\phi(\bar{x},\bar{y})$ and so for $\bar{n}\in B(k,\bar{m})$ \[||\phi^{r,\bar{x}}(\bar{m},\bar{n})||^M=||\phi(\bar{m},\bar{n})||^M=||\phi(\bar{m},\bar{n})||^{B(k,\bar{m})}.\]
    So suppose for $l\in\omega$ that for all formulas $\psi$ in prenex normal form with at most $l$ quantifiers that\[M^\ast\models[\phi(\bar{m},\bar{n})]^{r,\bar{x}}\text{ iff }B(r,\bar{m})^\ast\models \phi(\bar{m},\bar{n}).\]
    We consider $QzQ_1z_1...Q_lz_l(\phi(\bar{x},\bar{z},\bar{y}))$. Either $Q=\exists$ or $\forall$. The two cases are symmetric, we detail the $Q=\exists$ case.
    \[[\exists z Q_1z_1...Q_lz_l(\phi(\bar{x},\bar{z},\bar{y},z))]^{k,\bar{x}}=\exists z([Q_1z_1...Q_lz_l\phi(\bar{x},\bar{z},\bar{y},z)]^{r,\bar{x}}
    \land \delta_r(\bar{x},z)).\]
 Suppose $M^\ast$ models this formula. Then,
    \begin{multline*}
    1\leq ||\exists z([Q_1z_1...Q_lz_l\phi(\bar{m},\bar{z},\bar{n},z)]^{r,\bar{x}}\land \delta_k(\bar{x},z))||^M\\
    =||[Q_1z_1...Q_lz_l\phi(\bar{m},\bar{z},\bar{n},n)]^{r,\bar{x}}\land \delta_r(\bar{m},n)||^M.
    \end{multline*}
    for some $n\in M$ by witnessing. This implies
    \[||[Q_1z_1...Q_lz_l\phi(\bar{m},\bar{z},\bar{n},n)]^{r,\bar{x}}||^M
    \text{ and }
    ||\delta_r(\bar{m},n)||^M
    \text{ are both }\geq 1,\]
    in particular $n\in B(r,\bar{m})$. By the inductive hypothesis:
    \begin{align*}
    &||Q_1z_1...Q_lz_l\phi(\bar{m},\bar{z},\bar{n},n))||^{B(r,\bar{m})}\geq 1 \text{ and } n\in B(r,\bar{m)}\\
    \Rightarrow\ &||\exists z(Q_1z_1...Q_lz_l\phi(\bar{m},\bar{z},\bar{n},z))||^{B(r,\bar{m})}\geq 1.
    \end{align*}
    
    Conversely, if $||\exists z(Q_1z_1...Q_lz_l\phi(\bar{m},\bar{z},\bar{n},z))||^{B(r,\bar{m})}\geq 1$, then by witnessing we have $n\in B(r,\bar{m})$, i.e. $n\in M:||\delta_r(\bar{m},n)||\geq 1$ where \[||Q_1z_1...Q_lz_l\phi(\bar{z},\bar{n},n))||^{B(r,\bar{m})}\geq 1.\]
    By induction hypothesis $||[Q_1z_1...Q_lz_l\phi(\bar{m},\bar{z},\bar{n},z)]^{r,\bar{x}}||^M\geq 1$ which implies,
    \begin{align*}
      &||\exists z([Q_1z_1...Q_lz_l\phi(\bar{m},\bar{z},\bar{n},z)]^{r,\bar{x}}\land \delta_r(\bar{m},n))||^M\\
      &\qquad=
      ||[\exists z(Q_1z_1...Q_lz_l(\bar{m},\bar{z},\bar{n},z))]^{r,\bar{x}}||^M
      \geq 1.\qedhere
    \end{align*}
\end{proof}

Recall from remark \ref{rem:witnessedifflinear} that every finite model over an algebra $A$ is witnessed iff $A$ is linear. For this reason we restrict to the setting where $A$ is a residuated chain with co-atom and can give our many-valued understanding of local formula. 
\begin{defi}
    Let $A$ be a residuated chain and consider the distance formulas $\delta_r(x,y)$, $\rho_r(x,y)$ for any of the considered distance metrics. We define two special classes of formulas. 

    The local formulas are exactly those of the form $\phi^{r,\bar{x}}(\bar{x})$ with no free variables beyond those used to define the neighbourhood quantifiers. We call a sentence basic local iff it has the form \[\exists x_1...\exists x_n\bigland\limits_{1\leq i\not =j\leq n}\rho_{2r}(x_i,x_j)\land\bigland\limits_{1\leq i\leq n}\psi^{r,x_i}(x_i).\]
\end{defi}

We are now able to state and prove our version of Gaifman's Lemma. As a final useful shorthand, we say that a subset $X$ of a $\CP$-model is \emph{l-scattered} iff $\forall m,m'\in X\ d(m,m')>l$. 
\begin{thm}[Gaifman Lemma for Residuated Chains]\label{thm:GaifLem}
    Let $\CP$ be a finite relational language with object constants and let $M,N$ be finite $\CP$-models over a residuated chain $A$ with co-atom. Consider any encodable distance metric for $\CP$-models with distance formulas $\delta_r(x,y)$ and $\rho_r(x,y)$.

    \noindent
    Let $k\in\omega$, $\bar{m}\in M$ and $\bar{n}\in N$ and suppose that:\begin{enumerate}[label=\roman*.]
        \item For all local $\CP$-formulas $\phi^{7^k,\bar{x}}(\bar{x})\ M^\ast\models\phi^{7^k,\bar{x}}(\bar{m})$ iff $N^\ast\models\phi^{7^k,\bar{x}}(\bar{n})$.
        \item $M^\ast$ and $N^\ast$ satisfy the same basic local sentences.
    \end{enumerate}
    Then $M,\bar{m}\equiv^s_k N,\bar{n}$.
\end{thm}
\begin{proof}
    Once again we are able to mirror the textbook classical proof \cite[Theorem 1.5.1]{EbbinghausFlum95:Finmodeltheory}. The addition of a leading sequence requires some minor initial adjustments. We show that $\langle I_j\rangle_{j\leq k}$ is a b\&f-system for $M^\ast$ and $N^\ast$ where we define:\[I_j:=\{\bar{m}\bar{u}\mapsto\bar{n}\bar{v}\text{ p.iso}:|\bar{u}|\leq k-j\text{ and }(B(7^j,\bar{m}\bar{u}),\bar{m}\bar{u})^\ast\equiv^s_{g(j)}(B(7^j,\bar{n}\bar{v}),\bar{n}\bar{v})^\ast\}.\]
    where $g(j)$ is a sufficiently increasing function defined by constraints established in the course of the proof. The conclusion then follows by corollary \ref{cor:standardEF}. To ease notation we will simply write $\bar{u}\mapsto\bar{v}$ for a p.iso and assume the leading $\bar{m}\mapsto\bar{n}$ prefix.  

    We first claim that $\bar{m}\mapsto\bar{n}\in I_k$. Letting $\phi(\bar{x})$ be any $\CP^\ast\cup\{\bar{x}\}$-sentence, we observe:\[B(7^k,\bar{m})^\ast,\bar{m}\models\phi(\bar{m})\text{ iff }M^\ast\models\phi^{7^k,\bar{x}}(\bar{m})\text{ iff }N^\ast\models\phi^{7^k,\bar{x}}(\bar{n})\text{ iff }B(7^k,\bar{n})^\ast,\bar{n}\models\phi(\bar{n}),\]
    where the first and third bi-conditional holds by Lemma \ref{lem:localform} and the middle bi-conditional from assumption $i.$.

    By symmetry we just need to show the forth property holds for the system, that is for $0\leq j\leq k$, $u\in M$ and $\bar{u}\mapsto\bar{v}\in I_{j+1}$ we want to find $v\in N$ such that $\bar{u}u\mapsto\bar{v}v$ is a p.iso and $(B(7^j,\bar{u}u),\bar{u}u)^\ast\equiv^s_{g(j)} (B(7^j,\bar{v}v),\bar{v}v)^\ast$. We introduce a useful abbreviation for any tuple $\bar{d}$ belonging to any model $D$:\[\psi^j_{\bar{d}}(\bar{x}):=[\phi^{g(j)}_{B(7^j,\bar{d}),\bar{d}}(\bar{x})]^{7^j,\bar{x}}.\]
    
    Combining lemmas \ref{lem:localform} and corollary \ref{cor:standardEF} yields the important behaviour of the formula $\psi^j_{\bar{d}}(\bar{x})$ for any model $S$ and $\bar{s}\in S$: \[S^\ast\models\psi^j_{\bar{d}}(\bar{s})\text{ iff }B(7^j,\bar{s})^\ast\models\phi^{g(j)}_{B(7^j,\bar{d}),\bar{d})}(\bar{s})\text{ iff }B(7^j,\bar{s})^\ast,\bar{s}\equiv^s_{g(j)} B(7^j,\bar{d})^\ast,\bar{d}.\]

    \noindent
    \textbf{Case 1:}
    $u\in B(2\cdot 7^j,\bar{u})$. Then $u\in B(7^{j+1},\bar{u})$ and $B(7^{j+1},\bar{u})^\ast\models \exists z(\delta_{2\cdot7^j}(\bar{u},z)\land\psi^j_{\bar{u}u}(\bar{u}z))$. Provided $g(j+1)$ is large enough such that the quantifier rank of this formula is less than or equal to $g(j+1)$ we get by the inductive assumption:\[B(7^{j+1},\bar{v})^\ast,\bar{v}\models\exists z(\delta_{2\cdot 7^j}(\bar{v},z)\land \psi^j_{\bar{u}u}(\bar{v}z)),\] and by witnessing we get $v\in B(2\cdot 7^j,\bar{v})$ such that \[B(7^{j+1},\bar{v})^\ast,\bar{v}\models\psi^j_{\bar{u}u}(\bar{v}v),\text{ i.e. }B(7^j,\bar{u}u)^\ast,\bar{u}u\equiv^s_{g(j)} B(7^j,\bar{v}v)^\ast,\bar{v}v,\] therefore $\bar{u}u\mapsto\bar{v}v\in I_j$.

    \noindent
    \textbf{Case 2:}
    $u\not\in B(2\cdot 7^j,\bar{u})$, that is $B(7^j,\bar{u})\cap B(7^j,u)=\varnothing$. For $s\geq 1$ we define a formula $\gamma_s(\bar{x})$ which expresses that $\{\bar{x}\}$ is a $4\cdot 7^j$-scattered subset where every members' $7^j$ sphere has the same $g(j)$ isomorphism type as element of a model $u$.
    \[\gamma_s(\bar{x}):=\bigland_{1\leq r,t\leq s} \rho_{4\cdot7^j}(x_r,x_t)\land\bigland\limits_{1\leq r\leq s}\psi^j_u(x_r).\]
    We note that for any arbitrary model and tuple $\bar{d}\in D$, $D^\ast\models\gamma_s(\bar{d})$ iff for all $1\leq r,t\leq s\ D^\ast\models\rho_{4\cdot 7^j}(d_r,d_t)$ and $D^\ast\models\psi^j_u(d_r)$ iff $\{\bar{d}\}$ forms a $4\cdot 7^j$-scattered subset of $D$ and for each $d\in \bar{d}\ B(7^j,u)^\ast,u\equiv^s_{g(j)} B(7^j,d)^\ast,d$ iff $\{\bar{d}\}$ form a $4\cdot 7^j$-scattered subset and for each $d\in\bar{d}$ the $7^j\ a$-sphere of $d$ has the same $g(j)$ isomorphism type as $u$.

    We compare the relative cardinalities, defined as $e$ and $i$, of maximal subsets satisfying this in $M^\ast$ and $B(2\cdot 7^j,\bar{u})^\ast$ respectively. That is, we define $e$ and $i$ such that the following formulas hold:\begin{gather*}
        a.\, B(7^{j+1},\bar{u})^\ast\models\exists x_1...\exists x_e(\bigland\limits_{1\leq r\leq e} \delta_{2\cdot 7^j}(\bar{u},x_r)\land \gamma_e) \\
        b.\, B(7^{j+1},\bar{u})^\ast\not\models\exists x_1...\exists x_{e+1}(\bigland\limits_{1\leq r\leq e+1}\delta_{2\cdot 7^j}(\bar{u},x_r)\land \gamma_{e+1}) \\
        c.\, M^\ast\models\exists x_1...\exists x_i\gamma_i,\qquad d.\, M^\ast\not\models\exists x_1...\exists x_{i+1}\gamma_{i+1}
    \end{gather*}
    An easy consequence of witnessing establishes these represent the appropriate cardinalities. If no $i$ exists we set $i=\omega$; by contrast $e$ is bounded by the length of $\bar{u}$ (i.e. $k+|\bar{m}|$) since any two elements in the same sphere of radius $2\cdot 7^j$ have a distance at most $4\cdot 7^j$. Additionally, $e\leq i$ and we claim that $e$ and $i$ as determined in $N^\ast$ and $B(7^{j+1},\bar{v})$ are the same. In the case of $N^\ast$ this is because the sentences $\exists x_1...\exists x_i \gamma_i$ are basic $a$-local and so $M$ and $N$ are in mutual agreement to their validity. For $B(7^{j+1},\bar{v})^\ast$ we recall from the forth hypothesis that $(B(7^{j+1},\bar{u})^\ast,\bar{u})\equiv^s_{g(j+1)} (B(7^{j+1},\bar{v})^\ast,\bar{v})$.  Therefore, provided $g(j+1)$ is greater than the quantifier rank of sentences in a. and b. we have equality of $e$.

    \noindent
    \textbf{Case 2.1:}
    $e=i$. Then all elements satisfying $\psi^j_u(x)$ have $a$-distance $\leq 4\cdot 7^j+2\cdot 7^j=6\cdot 7^j< 7^{j+1}$ from $\bar{u}$. In particular, this includes $u$ itself, and then as $u\not\in B(2\cdot 7^{j+1},\bar{u})$ we have \[B(7^{j+1},\bar{u})^\ast\models\exists z(\delta_{6\cdot 7^j}(\bar{u},z)\land \rho_{2\cdot 7^j}(\bar{u},z)\land \psi^j_u(z)\land \psi^j_{\bar{u}}(\bar{u})),\]
    and provided $g(j+1)$ is greater than or equal to the quantifier rank of this sentence, by forth hypothesis \[B(7^{j+1},\bar{v})^\ast\models\exists z(\delta_{6\cdot 7^j}(\bar{u},z)\land \rho_{2\cdot 7^j}(\bar{u},z)\land \psi^j_u(z)\land \psi^j_{\bar{u}}(\bar{u})).\] Then (by witnessing) we have $v\in N$ such that: \begin{gather*}
        2\cdot 7^j<d(\bar{v},v)\leq 6\cdot 7^j\\
        B(7^j,u)^\ast,u\equiv^s_{g(j)} B(7^j,v)^\ast,v\\
        B(7^j,\bar{u})^\ast,\bar{u} \equiv^s_{g(j)} B(7^j,\bar{v})^\ast,\bar{v}
    \end{gather*}
    Moreover, the underlying sets are disjoint within their respective models and so we immediately have:\[B(7^j,\bar{u}u)^\ast,\bar{u}u\equiv^s_{g(j)} B(7^j,\bar{v}v)^\ast,\bar{v}v\]
    and therefore $\bar{u}u\mapsto\bar{v}v\in I_j$.

    \noindent
    \textbf{Case 2.2:}
    $e<i$. Now $N^\ast\models\exists x_1...\exists x_{e+1}\gamma_{e+1}$ so there is an element $v\in N$ such that $B(7^j,\bar{v})\cap B(7^j,v)=\varnothing$ and $N^\ast\models \psi^j_u(v)$, that is $(B(7^j,u)^\ast,u)\equiv^s_{g(j)} (B(7^j,v)^\ast,v)$ and we can proceed as in the previous case.

    It remains to demonstrate that we can define the function $g$. We do so recursively setting $g(0)=0$. Recall that our requirements for our function is that, for each $0\leq j\leq k$, $\bar{u}\mapsto \bar{v}\in I_{j+1}$ and $u\in M$, $g(j+1)$ is greater than or equal to the quantifier rank of the following sentences defined relative to $g(j)$:\begin{align*}
        &\exists z(\delta_{2\cdot 7^j}(\bar{u},z)\land\psi^j_{\bar{u}u}(\bar{u}z))\\
        &\exists x_1...\exists x_e(\bigland\limits_{1\leq r\leq e} \delta_{2\cdot 7^j}(\bar{u},x_r)\land \gamma_e)\\
        &\exists x_1...\exists x_{e+1}(\bigland\limits_{1\leq r\leq e+1}\delta_{2\cdot 7^j}(\bar{u},x_r)\land \gamma_{e+1})\\
        &\exists z(\delta_{6\cdot 7^j}(\bar{u},u)\land \rho_{2\cdot 7^j}(\bar{u},z)\land \psi^j_u(z)\land \psi^j_{\bar{u}}(\bar{u})).
    \end{align*}
    To that end, for each $0\leq 1\leq j$ and \emph{any} given tuple $(\bar{u},\bar{v})\in M^{k-(j+1)}\times N^{k-(j+1)}$ and $u\in M$, let $\#(\bar{u},\bar{v},u)$ be the maximum of the quantifier ranks of the above sentences. Define $\#(\bar{u},\bar{v},v)$ for each $v\in N$ symmetrically. As $M$ and $N$ are finite the set of all such numbers $\{\#(\bar{u},\bar{v},u),\#(\bar{v},\bar{u},v)\}$ is finite and has a maximum and we take $g(j+1)$ to be that maximum.
\end{proof}
\begin{rem}
    It is worth noting that our definitions of locality concepts requires the appeal to the standard expansions in the statement of the theorem. In the next section we will consider an alternative approach where the locality machinery is defined in directly. We will also see that in such contexts the appeal to standard expansions is still necessary (example \ref{exa:Gaiflem}).  
\end{rem}

\section{Locality in Semirings}\label{sec:semirings}
As mentioned earlier, locality in non-classical models has also been investigated by Bizi\`{e}re, Gr\"{a}del and Naaf (2023) who looked at models defined over ordered semirings. The details of the semiring setting, both in syntax and semantics, make a direct comparison somewhat difficult. Additionally, the authors focus directly on the potential existence of strongly equivalent Gaifman normal forms, i.e. the full Gaifman locality theorem, rather than the question of linking the modelling of local formulas and sentences to strong equivalence. Nevertheless, there are some comparisons we can make between our investigation and theirs which help illuminate the nature of locality in the two settings. First, let us briefly introduce ordered semiring semantics. 

\begin{defi}
    A commutative semiring is an algebra $(A,+,\cdot,0,1)$ with $0\not=1$ such that $(A,+,0)$ and $(A,\cdot,1)$ are commutative monoids, $\cdot$ distributes over $+$ and $0$ is an annihilator for $\cdot$. We say a semiring $K$ is \textit{naturally ordered} iff the relation $a\leq b$ iff $\exists c:a+c=b$ is a partial order. Two important subclasses are the lattice semirings where the $(A,\leq)$ is a bounded distributive lattice with $+=\lor$ and $\cdot=\land$ and the \textit{min-max} semirings where the natural order is a total order with $0$ minimum and $1$ maximum. 

    Given a relational predicate language $\CP$ and a naturally ordered semiring $A$ we define a $\CP$-structure $M$ as a triple $(M,\langle P^M\rangle_{P\in\CP},\langle \neg P^M\rangle_{P\in \CP})$ where:\begin{itemize}
        \item $M$ is a non-empty set,
        \item $P^M\colon M^n\rightarrow A$ for each $P\in\CP$
        \item $\neg P^M\colon M^n\rightarrow A$ for each $P\in\CP$ 
    \end{itemize}
    We use the term $\CP$-\textit{literals} to refer to both the usual and negated relation symbols applied to some $\bar{m}\in M$. We say a $\CP$-structure is \textit{model-defining} iff for any pair of complementary literals $P,\neg P$ precisely one of the values $P^M$ and $\neg P^M$ is 0. We say that it \textit{tracks only positive information} iff for all negative literals $\neg P^M\in\{0,1\}$. 
    Given a tuple $\bar{m}\in M$, we extend the interpretations of literals for $\bar{m}$ to a truth valuation $||-||^M$ for all \textit{classical} first-order formulas $\phi(\bar{x})$ written in negation normal form defined by induction:\begin{align*}
        ||L(\bar{m})||^M &= L^M(\bar{m}),& \\
        ||m=m||^M&= 1 & ||m=n||^M&=0 \text{ for }m\not=n \\
        ||\psi(\bar{m})\lor\lambda(\bar{m})||^M&=||\psi(\bar{m})||^M + ||\lambda(\bar{m})||^M & ||\psi(\bar{m})\land\lambda(\bar{m})||^M&=||\psi(\bar{m})||^M\cdot||\lambda(\bar{m})||^M\\
        ||\exists x\phi(\bar{m},x)||^M&=\sum\limits_{m\in M}||\phi(\bar{m},m)||^M & ||\forall x\phi(\bar{m},x)||^M&=\prod\limits_{m\in M}||\phi(\bar{m},m)||^M.
    \end{align*}
\end{defi}
\begin{rem}
    We can point out a number of differences in both syntax and semantics. On the semantic side, in the semiring setting the focus is more on understanding generalised versions of conjunction and disjunction with the existential and universal quantifiers understood as arbitrary indexed versions of these. This contrasts in our setting where we always take the quantifiers as infimum and supremum. In principle this does not require a conjunction or disjunction connective at all, although of course this is an essential part of the residuated lattice setting. The two understandings naturally align when working with lattice semirings, although it is important to recognise that the it is the lattice connectives $\land,\lor$ of a residuated lattice that form a lattice semiring, rather than the residuated conjunction $\cdot$.
    
    On the syntax side, the semiring setting only considers alternative interpretations of familiar classical formulas. By restricting to formulas written in negation normal forms the semantics can be set up without a true negation connective, which is instead applied only at the base atomic level. Formulas are then built on these literals rather than atomics.\footnote{We see an echo of this approach in our definition of isomorphism types. Much as negation is only applied to atomics and literals are taken as the union of these formulas we only apply the expanded truth constants to the base formulas $a\und R(\bar{x})$ and $R(\bar{x})\und a$, with formulas of the first kind `positive' and the second `negative'.} The relationship between positive and negative literals is then restricted by the requirements to be model defining and to track positive information. Together these restrictions mean that for any relation symbol $P$ and tuple $\bar{m}\in M$ \[ ||\neg P(\bar{m})||^M=\begin{cases}
        1 & \text{if }||P(\bar{m})||^M=0\\
        0 & \text{if }||P(\bar{m})||^M>0.
    \end{cases}\] 
    In the setting of residuated lattices negation is interpreted by the term  $\neg(x):=x\und 0$. Guaranteeing that all models defined over a residuated lattice are model defining and track positive information under this negation amounts to requiring that the lattice $A$ has no-zero divisors, i.e. $\forall a,b\in A\ a\cdot b=0$ implies $a=0$ or $b=0$. 

    Pulling this together, the two investigations overlap for $\CP$-models defined over residuated chains with no zero divisors, equivalently model-defining $\CP$-structures that track only positive information defined over min-max semirings, and where we are only concerned with the negation normal form fragment of the set of formulas.
\end{rem}

Even when we restrict attention to the overlap in subject matter, comparing the investigations requires care due to their definition of the concepts needed to discuss locality. The notion of Gaifman graph is essentially the same; two elements in a model are adjacent iff there is some relation symbol $R$ whose interpretation sends some tuple containing the elements to a value strictly greater than $0$ \cite[Definition 2.4]{BiziereGradelNaaf23:Localitysemiring}. This means distance and local neighbourhood are understood as we did in our positive recovery of Hanf locality and the corollary result by Fagin, Stockmeyer and Vardi - corollary \ref{cor:Hanf}. Indeed, Bizi\`{e}re, Gr\"{a}del and Naaf also prove a version of this locality result \cite[Theorem 3.6]{BiziereGradelNaaf23:Localitysemiring} which agrees with our own result at the point of overlap and is proved in essentially the same fashion. It is notable that both their result and corollary \ref{cor:Hanf} each apply in their own way to models with a more flexible semantics and richer syntax than that required to make a sensible comparison. In the ordered semiring case the result applies to models defined for all fully idempotent semirings, meaning it also applies to models where the quantifiers are not interpreted as supremum and infinimum respectively. Our result was for all well-connected bounded residuated lattices, meaning it applies even for the richer syntax including the residuated operators $\cdot,\und$ and $/$, and in particular does not require the nice behaviour of negation that is built into models defined over ordered semirings.

When it comes to Gaifman locality there is a sharper departure. This is driven by the more ambitious focus of Bizi\`{e}re, Gr\"{a}del and Naaf's investigation where they look to investigate the status of the full Gaifman theorem - that is whether every formula in their language is strongly equivalent to some Gaifman formula built only from local formula and basic local sentences. One consequence of this is that they require the associated formula $\phi^{r,\bar{x}}$ take on the \textit{exact} value of the original formula rather than simply matching in terms of the $\models$ relation \cite[Appendix A]{BiziereGradelNaaf23:Localitysemiring}. This leads them to directly introduce new \textit{ball quantifiers} of the form $Qy\in B^{\CP}_{r}(m)$ with the semantics \begin{align*}
    ||\exists y\in B^{\CP}_{r}(m)\phi(m,y)||^M&:=\sum\limits_{n\in B(r,m)}||\phi(m,n)||^M\\ 
    ||\forall y\in B^{\CP}_{r}(m)\phi(m,y)||^M&:=\prod\limits_{n\in B(r,m)}||\phi(m,n)||^M
\end{align*}
Similarly to our considerations of needing both positive and negative distance formulas they also directly introduce \textit{scattered quantifiers} which directly quantify over only elements sufficiently far away from the specified element. The different approaches of the two investigations highlight an important aspect of the Gaifman lemma; it links the $\models$ relation in standard expansions to equality of the $||-||^M$ function. This link is provided by the residuated connectives $\und,/$ which interpret the order. Without that connective we lose the ability to prove the same Gaifman lemma when restricting to the negation normal fragment. One could ask about the possibility of a weaker Gaifman lemma where we only ask for equivalence at the level of $\models$, i.e. that $M,\bar{m}\equiv_k N,\bar{m}$. This does in fact hold, following directly from the classical Gaifman lemma.

\begin{cor}
    Let $\CP$ be a finite relational language with object constants and let $M,N$ be finite $\CP$-models over a min-max semiring $A$. We expand our language with ball and scattered quantifiers and define local formula and basic local sentences accordingly. Suppose that $M$ and $N$ satisfy the same basic local sentences. Then $M\equiv N$. 
\end{cor}
\begin{proof}
    For each semiring $\CP$-model $M$ we define a classical $\CP$-model counterpart $M^c$ to have the same underlying set and where for each relation symbol $R\in\CP$ we set:\[R^{M^c}:=\{\bar{m}\in M:R^M(\bar{m})>0\}.\]
    By a straightforward induction one can check for any negation normal formula $\phi(\bar{x})$, possibly including the additional ball and scattered quantifiers, that \[M\models\phi(\bar{m})\text{ iff }M^c\models\phi(\bar{m}).\]
    In particular, $M$ and $N$ satisfying the same basic local sentences implies $M^c$ and $N^c$ satisfy the same basic local sentences. We note that the semantics of the added quantifiers is equivalent to the standard notion \cite[Definition 4.1]{BiziereGradelNaaf23:Localitysemiring},  therefore by the classical Gaifman lemma $M^c$ and $N^c$ are logically equivalent and so $M$ and $N$ are as well. 
\end{proof}

That equivalence in modelling for the negation normal fragments with the `flattened' classical model can also be used to show that our lemma with the strong equivalence conclusions can fail in the setting of min-max semirings.
\begin{exa}\label{exa:Gaiflem}
    Consider a predicate language with a single monadic predicate $P$ and the min-max semiring defined on the unit interval $[0,1]$. We define the models $M$ and $N$ whose domain is a single element $s$ with\begin{align*}
        P^M(s)=\frac{3}{4} && P^N(s)=\frac{1}{2}.
    \end{align*}
    On the one hand $||\exists xP(x)||^M=\frac{3}{4}$ and $||\exists xP(x)||^N=\frac{1}{2}$ so the models are not strongly equivalent. However, the classical models $M^{\ast^c}$ and $N^{\ast^c}$, that is the classical counterparts to the standard expansions in the language $\CP^\ast$ are identical. Therefore $M^{\ast^c}$ and $N^{\ast^c}$ certainly satisfy the same basic local sentences and so by the same induction as in our previous corollary $M^\ast$ and $N^\ast$ satisfy the same basic local sentences. 

    We can expand our algebraic language to include residuated connectives. Recall from example \ref{exa:wellconnectboundresi} the standard G\"{o}del chain on the unit interval $[0,1]_G$ where we interpret the (commutative) residuaed conjunction $\cdot$ and its residuum $\rightarrow$ by: \begin{align*}
        a\cdot b=min\{a,b\} && a\rightarrow b=\begin{cases}
            1 & \text{if }a\leq b,\\
            b & \text{otherwise}.
        \end{cases}
    \end{align*}
    Defining the models $M$ and $N$ as above, we can show that the appeal to the behaviour of standard expansions is necessary in the Gaifman lemma, even with a coherent notion of local formulas and basic local sentences provided by the defined ball and scattered quantifiers. Just as before the two models $M$ and $N$ are not strongly equivalent. Now, in \cite[Lemma 11, Example 12]{DellundeGarciaNoguera18:Fuzzybackandforth} we find an example of two models similarly defined to ours which are shown to be logically equivalent. Following the same reasoning (with some straightforward adjustments) one can check that our models are logically equivalent even for the semantics with the added ball and scattered quantifiers. In particular, they must satisfy the same basic local sentences.  
\end{exa}

Returning to Bizi\`{e}re, Gr\"{a}del and Naaf's own work, their main result is a recovery of the full Gaifman theorem for negation normal sentences for finite models defined over min-max ordered semirings.
\begin{thmC}[\cite{BiziereGradelNaaf23:Localitysemiring}]\label{thm:Gaifsemiring}
\emph{(Gaifman normal forms in min-max semirings)} Let $\CP$ be a finite relational signature. For every negation normal sentence $\psi$ there is a local sentence $\lambda$ such that for every min-max semiring $A$ and for each finite $\CP$-model $M$ over $A\ ||\psi||^M=||\lambda||^M$.   
\end{thmC}
\noindent
It is notable that we see a restriction to min-max semirings, i.e.\ semirings whose underlying order is linear just as we saw a restriction to residuated chains for our Gaifman lemma. Whilst there is no analogous condition to the presence of a co-atom, semiring semantics builds a similar kind of restriction to negation with the idea of model-defining, positive information tracking structures. Despite this, as already discussed one must be careful when making comparisons. For example, the $n$-element MV-chains $\Luk_n$ introduced in example \ref{exa:resichaincoatom} are residuated chains with co-atoms and so fall under the purview of our Gaifman lemma. Both their $[\lor,\land]$ and $[\lor,\cdot]$-reducts are ordered semirings, but only the former are min-max semirings to which the preceding result applies.

The proof strategy to theorem \ref{thm:Gaifsemiring} is similar to Gaifman's original proof - a constructive elimination of alternative quantifiers - although with fairly significant adjustments to account for the restriction to sentences and the particular behaviour of negation in the ordered semiring setting. It is highly reliant on the restriction to negation normal formulas, specifically that beyond the base level of literals the only connectives are $\land$ and $\lor$ which are both monotone with respect to the order. This is not true of the residuated connectives $\und$ and $/$ posing a significant barrier to adapting the quantifier elimination proof strategy to the setting of the residuated chains with co-atoms for which we recovered our lemma. Whether one can indeed use this proof strategy to recover a full Gaifman theorem remains an open question deserving of further research.

\section{Queries}\label{sec:queries}
After our proof of the initial Hanf's theorem we remarked that our ability to immediately translate the classical proof to the many-valued setting indicated that our interest should be in these theorems antecedent. Our Gaifman result, where we generalised not the usual Gaifman's theorem but its main lemma, somewhat reinforces this; in both cases what we obtained is a sufficiency condition for models being strongly elementary equivalent. Our interest then turns to what it means for model pairs to satisfy the condition and how that might look different in a many-valued context. Accordingly, we turn to a topic that both utilises these theorems and demonstrates a few standard model comparisons that do meet these conditions - model queries. 

\begin{defi}[Model Query]
    An $n$-ary query for $n\geq 0$ on $\CP$-models is a mapping $Q$ that associates a $\CP$-model $(A,M)$ with a subset of $M^n$ closed under isomorphism, i.e. if $(A,M)\isom (B,N)$ via $(f,g)\colon (A,M)\rightarrow (B,N)$ then $Q(N)=f[Q(M)]$.

    We say that $Q$ is definable (in a logic of type $\CL$) iff there is an $\CP$-formula $\phi(\bar{x})$ such that for every $(A,M)$: 
    \[Q(A,M)=\{\bar{m}\in M^n:(A,M)\models\phi(\bar{m})\}.\]
\end{defi}

\begin{rem}
    If we fix our algebra $A$ we can consider queries definable in the language $\CP^\ast$ defined relative to standard models. That is queries $Q$ for which there is a $\CP^\ast$ formula $\phi(\bar{x})$ such that for every $\CP$-model $M$: \[Q(M)=\{\bar{m}\in M^n:M^\ast\models\phi(\bar{m})\}.\]
    Moreover, for each $a\in A$ we can consider the notion of $a$-definable queries where $Q$ is $a$-definable iff there is an $\CP^\ast$-formula $\phi(\bar{x})$ such that for every $\CP$-model $M$: \[Q(M):=\{\bar{m}\in M^n:||\phi(\bar{m})||^{M}\geq a\}.\]
    In the context of standard models these definitions are equivalent in the sense that a query is definable iff for all $a\in A$ it is $a$-definable iff it is $a$-definable for some $a\in A$. This follows from the syntactic translations $\phi\cdot a$ and $a\und \phi$ respectively.

    We could further consider the idea of strict $a$-definable queries where $Q$ is  strict $a$-definable iff there is an $\CP^\ast$-formula $\phi(\bar{x})$ such that for every $\CP$-model $M$:\[Q(M)\coloneqq \{\bar{m}\in M^n:||\phi(\bar{m})||^M> a\}.\]
    When $A$ is a residuated chain with co-atom, as is the setting for our Gaifman lemma, the presence of the co-atom lets us use the easy translation of $(\phi\und a)\und c$ where $c$ is the constant symbol corresponding to the co-atom in the algebra. In the setting of our Hanf locality result, bounded well-connected residuated lattices, we do not have the same easy syntactic translations to give an equivalence with strict $a$-definable queries and definable queries.
\end{rem}

Given our primary interest in models defined over a fixed algebra we will proceed by looking at queries definable relative to standard models. When $n=0$ we use $1$ and $0$ in place of $M^0$ and $\varnothing$ for the possible outputs of a query. We can formulate two locality criteria for queries that align with our two locality theorems. Recall the relation $M\leftrightarrows_r N$ stating that for each isomorphism type $\tau$ of $\CP$-models the number of elements of $M$ and $N$ with $r$-sphere type $\tau$ is the same, defined relative to any (strict) threshold distance metric. 
\begin{defi}
    Let $d(x,y)$ be a distance metric for $\CP$-models. An $n$-ary query $Q$ on $\CP$-models is $d$-Hanf local iff there is a number $r\geq 0$ such that for every pair of $\CP$-models $M,N\ \bar{m}\in M^n$ and $\bar{n}\in N^n$: \[(M,\bar{m})\leftrightarrows_{r} (N,\bar{n}) \text{ implies }\bar{m}\in Q(M)\text{ iff }\bar{n}\in Q(N).\]
    The smallest $r$ for which the above condition holds is called the Hanf locality rank of $Q$ and are denoted by $hr(Q)$.
\end{defi}

\begin{defi}
    Let $d(x,y)$ be a distance metric for $\CP$-models. An $n$-ary query $Q$, $n>0$ on $\CP$-models is called $d$-Gaifman local iff there exists a number $r\geq 0$ such that for every $\CP$-model $M$ and $\bar{m_1},\bar{m_2}\in M^n$: \[B(r,\bar{m_1})\isom B(r,\bar{m_2})\Rightarrow \bar{m_1}\in Q(M)\text{ iff }\bar{m_2}\in Q(M).\]
\end{defi}
Our two forms of locality each provide a method for proving the inexpressibility of a query. We can establish a corollary to each locality theorem that definable queries are Hanf/Gaifman local respectively. Thus, to show a query is not definable it is sufficient to check it is not Hanf/Gaifman local. Of course this is only applicable for the settings where we recovered the respective locality result. Note a key difference between these two is that Hanf local queries are concerned with different structures whilst Gaifman local ones are concerned with a single structure. For this reason Hanf locality tends to be most useful for $0$-ary queries and Gaifman locality when $n>0$. 

\begin{cor}
    Let $Q$ be an $n$-ary query on $\CP$-models defined over a residuated chain $A$ with co-atom, and $d(x,y)$ a (strict) threshold distance metric. Every definable query is $d$-Gaifman local.
\end{cor}
\begin{proof}
    Let $Q$ be defined by the $\CP^\ast$-formula $\phi(\bar{x})$. Let $k=qr(\phi)$ and $d=g(k)$ where $g$ is the function defined in the proof of the residuated chain Gaifman lemma. Suppose $B(d,\bar{m})\isom B(d,\bar{n})$, we note this extends to an isomorphism $B(d,\bar{m})^\ast\isom B(d,\bar{n})^\ast$. $M^\ast$ trivially satisfies the same basic local sentences as itself. Let $\phi^{k,\bar{x}}(\bar{x})$ be the local $\CP^\ast$-formula induced by $\phi$. Then $M^\ast\models\phi^{k,\bar{x}}(\bar{m})$ iff $B(k,\bar{m})^\ast\models\phi(\bar{m})$ iff $B(k,\bar{n})^\ast\models\phi(\bar{n})$ iff $M\models\phi^{k,\bar{x}}(\bar{n})$. Therefore by the Many-valued Gaifman's lemma $(M,\bar{m})\equiv^k_s (M,\bar{n})$ and: \[\bar{m}\in Q(M) \text{ iff }M^\ast\models\phi(\bar{m}) \text{ iff }M^\ast\models\phi(\bar{n}) \text{ iff }\bar{n}\in Q(M).\qedhere\]
\end{proof}

\begin{cor}
    Let $Q$ be an $n$-ary query on $\CP$-models defined over a well-connected bounded residuated lattice $A$. Let $d_{>\bot}$ be the strict $\bot$ threshold metric. Then:\begin{itemize}
        \item Every definable query is $d_{>\bot}$-Hanf local.
        \item For all $a\in A$, every $a$-definable query is $d_{>\bot}$-Hanf local.
        \item For all $a\in A$, every strict $a$-definable query is $d_{>\bot}$-Hanf local.
    \end{itemize}
\end{cor}
\begin{proof}
    Let $Q$ be a definable query by $\phi(\bar{x})$. Letting $k=qr(\phi)$ we consider $d=3^k$ and suppose $(M,\bar{m})\leftrightarrows_d (N,\bar{n})$. From corollary \ref{cor:Hanf} we have $(M,\bar{m})\equiv^k_s(N,\bar{n})$ and so $\bar{m}\in Q(M)$ iff $M\models\phi(\bar{m})$ iff $N\models\phi(\bar{n})$ iff $\bar{n}\in Q(N)$. The other claims follow similarly simply noting that $(M,\bar{m})\equiv^k_s (N,\bar{n})$ also implies $||\phi(\bar{m})||^M\geq a$ iff $||\phi(\bar{n})||^N\geq a$ and $||\phi(\bar{m})||^M> a$ iff $||\phi(\bar{n})||^N>a$.
\end{proof}
\begin{rem}
    We can easily give a definition of local queries that utilises the full power of Hanf locality. We define a query $Q$ to be \textit{threshold local} iff $\exists k\geq 0$ such that for all $\CP$-models $M,N$ and $\bar{m}\in M$, $\bar{n}\in N$ if for each $r$-sphere type $\iota$ where $r\leq k$, $M,\bar{m}$ and $N,\bar{n}$ have either the same number of points of $r$-sphere type $\iota$ or at least $k\cdot e$ elements of $r$-sphere type $\iota$ (where $e$ is the maximum size of $k$-spheres in $M$ and $N$), then $\bar{m}\in Q(M)$ iff $\bar{n}\in Q(N)$.
\end{rem}

\subsection{Examples of locality}
As a small flavour of the method in example, we present a many-valued variant of the two standard queries for which this technique is demonstrated in the classical case \cite[Chapter 2]{Gradeletal07:Finitemodelapplication}. We work in the theory of undirected weighted graphs, that is our signature is a single binary relation symbol $E$ 

\begin{exa}
    We say that a weighted graph defined over an algebra $A$ is $t$-connected iff for any pair of vertices $x,y,$ there is a path in the graph from $x$ to $y$ where every edge has weight $>t$. The $\bot$-connectivity query $Q$ is the $0$-ary query for weighted graphs defined over a well-connected bounded lattice $A$ defined by: \[Q(M)=\begin{cases}
        1 & \text{ if  M is $\bot$-connected}\\
        0 & \text{ o.w.}
    \end{cases}\] 
    This query is not definable for finite weighted graphs. 

    We consider $M$ to be a cycle of length $2m$ where each edge has weight $>\bot$ and all the remaining edges have weight $\bot$. We let $N$ be two $m$ cycles where each edge has weight $>\bot$ and all remaining edges weight $\bot$ where $m>2r+1$ for arbitrary $r\in\omega$.

    Because each cycle has length $>2r+1$ for any element of $M\cup N\ B_{>\bot}(r,x)$ is always a weighted chain of length $2r+1$ where each edge has weight $>\bot$ and all others are weight $\bot$. Therefore every element of $M\cup N$ realises the same $r$-sphere type and so $M\leftrightarrows_r N$ but $Q$ does not agree on $M$ and $N$ so $Q$ is not $d_{>\bot}$-Hanf local.
\end{exa}

\begin{exa}
    Given two elements $m,n$ belonging to a weighted graph defined over an algebra $A$, we say that $n$ is in the $t$-transitive closure of $m$ iff there is a path from $m$ to $n$ where each edge has weight $\geq t$. The $t$-transitive closure query is the $2$-ary query for weighted graphs defined over a residuated chain $A$ with co-atom defined by: \[Q(M)=\{(m,n)\in M^2:n\text{ is in the }t\text{-transitive closure of }m\}.\]
    This query is not definable for finite weighted graphs.
    Let $r\in\omega$ and consider $M$ to be a \textit{directed} $t$-chain of length $4r+4$, that is the weighted graph on $4r+4$ elements with a chain of edges all of weight at least $t$ and all other edges of weight $<t$. Let $m$ be the element in position $r+1$, and $n$ the element in the $3r+3$ position. Then $B_t(r,(m,n))\isom B_t(r,(n,m))$ as each is just the disjoint union of $2$ directed $t$-chains of length $2r+1$ centred on $m$ and $n$ respectively. However, $(m,n)\in Q(M)$ and $(n,m)\not\in Q(M)$ so $Q$ is not $d_{t}$-Gaifman local.
\end{exa}


\section{Conclusions and Further Study}\label{Conclusion}
Here we have begun the investigation on locality for finite residuated lattice models. Following the usual pattern for work in non-classical model theory we aimed to define appropriate generalisations of classical notions and ask whether the classical results still hold. Starting with Hanf locality, we found that for the most natural definition of distance in a model the theorem failed in all algebras beyond the 2-element Boolean algebra. Nevertheless, there was one definition of distance for which the theorem was recovered for models defined over any well-connected bounded residuated lattice. When considering Gaifman locality we had to restrict to the much better behaved class of residuated chains with co-atom in order to even define the concept of local formula and basic local sentence it is concerned with. In contrast to Hanf locality however, within that setting any of the considered distance metrics were applicable meaning the result could be recovered even for models defined over unbounded algebras. A significant tool throughout were the $k$-isomorphism types, a syntactic encoding of back-and-forth systems made possible by our residuated connectives $\und,/$.

There remain many open questions. Perhaps the most pressing would be to better understand to what extent the restrictions we required to recover the theorems are necessary and to seek potential counterexamples to the Gaifman lemma for models defined over residuated lattices that are not well-connected or unwitnessed models generally. Doing so faces significant challenges, that all models considered are witnessed is woven into the behaviour of every use of syntax to encode some semantic property that relies on quantification. Perhaps more significantly, in the case of relativised quantifiers we have an encoding dependent on witnessing that is built into the definition of the local formula and basic local sentences. As it stands the very statement of Gaifman's lemma is only coherent for witnessed models. During our comparison to locality for semiring semantics we highlighted the question of whether the full Gaifman theorem might be recovered by following the proof strategy based on quantifier elimination utilised there. The presence of residuated connectives poses a serious barrier to adapting the proof methods used in \cite{BiziereGradelNaaf23:Localitysemiring}, as unlike the min-max operations they are not monotone in both arguments. Accordingly, a more modest possibility would be to consider a syntax restriction and focus only on so-called `positive' sentences defined with similar inspiration as they are in classical model theory. 

Notably, in this investigation we focused solely on finite models. Whilst a natural setting for a locality investigation given the relative importance of locality results in finite model theory \cite{EbbinghausFlum95:Finmodeltheory,Libkin04:Elementsfinmod}, locality concepts can certainly be sensibly applied to the potentially infinite case. Indeed, both the classical Hanf \cite{Hanf65:Modeltheory} and Gaifman \cite{Gaifman81:Local} results were proved for arbitrary models, the natural question then is whether a similar investigation could be carried out for arbitrary models and how critical our restriction to finite models is in the preceding work. An immediate cause for concern is that for arbitrary models the $k$-isomorphism types are no longer well-defined, requiring an infinite conjunction/disjunction. Classically, this is no issue as the set of formulas in a bounded number of variables is finite up to logical equivalence, but an equivalent result has not been established in the generalised setting. This question is additionally interesting as it links to understanding the situations in which the reverse direction of the Ehrenfeucht-Fra\"{i}ss\'{e} theorem can be recovered, something not guaranteed in the residuated lattice setting \cite[Example 23]{DellundeGarciaNoguera18:Fuzzybackandforth}. Alternatively, one could look to quantifier elimination methods as used by both Gaifman in his original classical proof and Bizi\`{e}re, Gr\"{a}del and Naaf for semirings.

\section*{Acknowledgment}
We were partially supported by the Australian Research Council grant DE220100544.

\bibliographystyle{alphaurl}
\bibliography{ref}

\end{document}